\documentclass[aps,pra,notitlepage,onecolumn,superscriptaddress]{revtex4-2}

\usepackage{amsfonts}
\usepackage{amsmath}
\usepackage{amsthm}
\usepackage{bm}
\usepackage[colorlinks,linkcolor=blue,anchorcolor=blue,citecolor=blue,urlcolor=blue]{hyperref}
\usepackage[capitalize]{cleveref}
\usepackage{enumerate}
\usepackage[tmargin=1in, bmargin=1in, lmargin=1.25in, rmargin=1.25in]{geometry}
\usepackage{graphicx}
\usepackage[utf8]{inputenc}
\usepackage{mathtools}
\usepackage{tablefootnote}
\usepackage{threeparttable}
\usepackage{thmtools}
\usepackage{thm-restate}
\usepackage{svg}

\graphicspath{ {./figures/} }

\theoremstyle{plain}
\newtheorem{theorem}{Theorem}

\theoremstyle{definition}
\newtheorem{definition}{Definition}

\newtheorem{lemma}[definition]{Lemma}

\newtheorem{remark}[definition]{Remark}

\DeclarePairedDelimiter{\abs}{|}{|}

\DeclarePairedDelimiter{\bracks}{[}{]}

\DeclarePairedDelimiter{\parens}{(}{)}

\def\maj#1{\textsc{maj-}$#1$}
\newcommand{\nand}{\textsc{nand}}
\newcommand{\xnand}{\textsc{xnand}}
\newcommand{\den}{\textsc{den}}
\newcommand{\projone}{\textsc{projlsb}}
\newcommand{\projtwo}{\textsc{projmsb}}
\newcommand{\lift}{\textsc{lift}}
\newcommand{\errornand}{\textsc{enand}}

\begin{document}

\title{Reliable computation by large-alphabet formulas in the presence of noise}

\author{Andrew K. Tan}
  \email{aktan@mit.edu}
  \affiliation{Department of Physics, Co-Design Center for Quantum Advantage, Massachusetts Institute of Technology, Cambridge, Massachusetts 02139, USA}
\author{Matthew Ho}
  \email{mattho@mit.edu}
  \affiliation{Department of Mathematics, Massachusetts Institute of Technology, Cambridge, Massachusetts 02139, USA}
\author{Isaac L. Chuang}
  \affiliation{Department of Physics, Co-Design Center for Quantum Advantage, Massachusetts Institute of Technology, Cambridge, Massachusetts 02139, USA}
  \affiliation{Department of Electrical Engineering and Computer Science, Massachusetts Institute of Technology, Cambridge, Massachusetts 02139, USA}

\date{\today}

\begin{abstract}
  We present two new positive results for reliable computation using formulas over physical alphabets of size \(q > 2\).
  First, we show that for logical alphabets of size \(\ell = q\) the threshold for denoising using gates subject to \(q\)-ary symmetric noise with error probability \(\varepsilon\) is strictly larger than that for Boolean computation, and we show that reliable computation is possible as long as signals remain distinguishable, i.e. \(\epsilon < (q - 1) / q\), in the limit of large fan-in \(k \rightarrow \infty\).
  We also determine the point at which generalized majority gates with bounded fan-in fail, and show in particular that reliable computation is possible for \(\epsilon < (q - 1) / (q (q + 1))\) in the case of \(q\) prime and fan-in \(k = 3\).
  Secondly, we provide an example where \(\ell < q\), showing that reliable Boolean computation, \(\ell = 2\), can be performed using \(2\)-input ternary, \(q = 3\), logic gates subject to symmetric ternary noise of strength \(\varepsilon < 1/6\) by using the additional alphabet element for error signaling.
\end{abstract}


\maketitle
\vspace{-2em}


\section{Introduction}
\label{sec:intro}
  The problem of performing computation with noisy components was first studied by von Neumann in 1952 \cite{von-neumann1956probabilistic}.
  Motivated by understanding the robustness of information processing in biological systems, he proposed a toy model of Boolean computation using noisy circuits in which Boolean-valued functions were computed through composition of basic gates with bounded fan-in.
  Interestingly, von Neumann showed that this model admitted a fault-tolerant regime: for circuits with noise parameterized by strength \(\varepsilon\), there exists a threshold \(\beta\) for which noisy circuits composed of gates subject to error \(\varepsilon < \beta\) could simulate a noiseless circuit of size \(N\) with probability of error bounded by \(\varepsilon_{\text{L}}\) with a modest overhead in the number of gates.
  This overhead, a multiplicative factor of \(\Theta(\log(N / \varepsilon_{\text{L}}))\) was argued by von Neumann \cite{von-neumann1956probabilistic} and later made rigorous by Dobrushin and Ortyukov \cite{dobrushin1977lower}.
  The key observation made by von Neumann was that by encoding data in an error correcting code and interleaving computation with a constant-depth error correction circuit, the error of the Boolean circuit's output could be bounded away from \(1/2\) in a manner independent of circuit depth; such a circuit computes a function \emph{reliably}, in a way which can be made mathematically precise.
  A more modern statement of this result allows for the possibility of an error correction circuit whose depth grows with error rate leading to a more general, polylogarithmic, overhead \cite{nielsen2010quantum}.
  The positive result of \cite{von-neumann1956probabilistic} has since been elegantly tightened for formulas (i.e. circuits of fan-out one) \cite{hajek1991on-the-maximum,evans1998on-the-maximum,evans2003on-the-maximum}, as summarized in \cref{tab:positive-result-summary}.
  Despite the use of formulas, a qualitatively different setting than circuits, these positive results for formulas nonetheless follow the essence of von Neumann's original construction, whose hallmark is the interleaving of distinct computation and error correction stages.

  Given a model of noisy computation, i.e. a set of elementary operations over a finite alphabet along with a specification of the error process,  the \emph{von Neumann construction} for noisy formula-based fault-tolerance may be distilled into the following two-step process:
  \begin{enumerate}
    \item
    Show that a given denoising operation (e.g. majority) has \(\ell \ge 2\) stable fixed-points in the probability simplex over the physical alphabet up to \(\varepsilon < \beta'\) where \(\varepsilon\) parameterizes the error model:
    each fixed-point along with its basin of attraction under this denoising operation, is associated with a logical state.
    \item
    Exhibit a set of gates capable of performing operations on the logical states, maintaining their outputs in the correct basins of attraction for \(\varepsilon < \beta \le \beta'\).
  \end{enumerate}
  We call \(\beta'\) the \emph{denoising threshold} of this particular construction, and \(\beta\) its \emph{computation threshold}.

  Though von Neumann's original analysis was limited to the toy model of noisy Boolean circuits, inspired by biological systems, he left open the possibility for reliable computation over non-binary alphabets, which have found renewed interest in certain modern computational models.
  Specifically, multivalued logic \cite{smith1981the-prospects,smith1988a-multiple} has recently been studied both for providing more natural embeddings for certain computational problems \cite{epstein1974the-development,saffiotti1995a-multivalued,dubrova1999multiple-valued,rine2014computer}, and for admitting efficient implementations in new materials \cite{morisue1989a-novel,dubrova1999multiple-valued,jo2021recent}.
  More recently still, multi-level quantum systems, i.e. qudits, have been shown to offer promising avenues (beyond qubits) for fault-tolerant quantum information processing \cite{bocharov2017factoring,wang2020qudits}.

  Motivated by moving beyond Boolean alphabets, we make initial steps towards extending formal positive fault-tolerance results using the von Neumann construction for computation over larger alphabets.
  Formally, we define an \emph{alphabet} of size \(q \ge 2\) to be the set \([q] \equiv \{0, 1, \dots, q - 1\}\);
  and we refer to elements of the alphabet set, either physical or logical, as \emph{characters}.
  A number of qualitative differences emerge in the study of larger alphabet fault-tolerance, with two specific opportunities (and matching challenges) arising.

  First is the opportunity allowed by larger alphabets to recover from higher error rates, compared with the case for a binary alphabet.
  This is due to the simple fact that any particular error is increasingly less likely to overwhelm the signal as \(q\) grows.
  However, it is not manifest that such error correctability extends to allow reliable \emph{computation} over larger alphabets for gates with bounded fan-in.

  Second, when \(q > 2\), logical computation may be performed without using the entire alphabet --- in other words, the logical alphabet size \(\ell\) can be strictly smaller than the physical alphabet size \(q\). 
  As a result, members of the physical alphabet that are not part of the logical alphabet can be used to signal that an error has occurred.  Whether this opportunity for ``error signaling'' can improve fault-tolerance thresholds, however, has never before been rigorously established.

  
  \begin{table}
  \centering
    \begin{threeparttable}
    \caption{Overview of positive fault-tolerance results for gates subject to binary, i.e. \(q = 2\) symmetric noise and our extensions to \(q > 2\).
    The first positive results were shown for circuits by von Neumann, who analyzed computation with both \(\{\maj{[2, 3]}, \nand\}\) and \(\{\nand\}\) gate sets \cite{von-neumann1956probabilistic}.
    These results were subsequently tightened: first for the \(\{\maj{[2, 3]}, \nand\}\) gate set by Hajek and Weller \cite{hajek1991on-the-maximum}, and for the more challenging fan-in \(2\) gate set \(\{\nand\}\) by Evans and Pippenger \cite{evans1998on-the-maximum}.
    Evans and Schulman demonstrated tight thresholds for \(\{\maj{[2, k]}, \xnand\}\), \(k\) odd \cite{evans2003on-the-maximum}.
    We derive results for larger \(k\) and \(q\) where \(C^{[q, k]}\) are defined in \cref{eq:big-C-def}.
    We also analyze an alternate scheme where some alphabet elements are used to signal errors assuming access to \den~and \errornand~gates defined in \cref{sec:universal-ternary-computation}. Other gates also given in \cref{sec:universal-ternary-computation} are used to project to and lift from a binary logical alphabet, but are not used for computation or denoising in this scheme.
}
    \label{tab:positive-result-summary}
      \begin{tabular}{|l|l|l|l|l|}
      \hline
      Denoising Gate & Denoising Threshold & Computation Gate(s) & Computation Threshold & Reference \\ 
      \hline\hline
      \maj{[2, 3]} & 1/6 & \nand & \(\approx 0.0073\) & \cite{von-neumann1956probabilistic} \\ \hline
      \nand & - & \nand\tnote{1} & \(\approx 0.0107\) & \cite{von-neumann1956probabilistic} \\ \hline
      \maj{[2, 3]} & \(1/6\)  & \xnand & \(1/6\) & \cite{hajek1991on-the-maximum}\tnote{2} \\ \hline
      \nand & \((3 - \sqrt{7})/4\) & \nand & \((3 - \sqrt{7})/4\) & \cite{evans1998on-the-maximum}\tnote{2} \\ \hline
      \maj{[2, k]}, \(k\) odd & \(\left.\frac{1}{2} - 2^{k-2} \middle/ k \binom{k-1}{\frac{k-1}{2}} \right.\) & \xnand & \(\left.\frac{1}{2} - 2^{k-2} \middle/ k \binom{k-1}{\frac{k-1}{2}} \right.\) & \cite{evans2003on-the-maximum}\tnote{2} \\ \hline
      \maj{[q, k]}, \(k\) odd & \(\frac{q-1}{q} \frac{C^{[q, k]} - 1}{C^{[q, k]}}\) & \(\textsc{add}^q\), \(\textsc{mul}^q\) & \(\frac{q-1}{q} \frac{C^{[q, k]} - 1}{C^{[q, k]}}\)~\tnote{3} & \cref{ssec:symmetric-computation-threshold} \\ \hline
      \maj{[q, 3]}, \(q\) prime & \(\frac{q - 1}{q (q + 1)}\) & \(\textsc{add}^q\), \(\textsc{mul}^q\) & \(\frac{q - 1}{q (q+ 1)}\) & \cref{lem:symmetric-computation-threshold-example} \\ \hline
      \(\den\) & 1/6 & \errornand\tnote{2} & 1/6 & \cref{sec:universal-ternary-computation} \\ \hline

      \end{tabular}
    \begin{tablenotes}
    \item[1] Result is for circuits.
    \item[2] Universal for binary computation.
    \item[3] Computation threshold is asymptotic for \(k \rightarrow \infty\) but conjectured to hold generally.
    \end{tablenotes}
    \end{threeparttable}
  \end{table}


  Here, we perform an analysis of the thresholds for reliable computation using gates subject to \(q\)-ary symmetric noise of strength \(\varepsilon\), i.e. where an error occurs with probability \(\varepsilon\), leaving the output in any of the \(q - 1\) erroneous states with equal probability.
  We study computation over physical alphabets of size \(q > 2\) using the two step von Neumann construction.
  We consider two main settings: one where \(\ell = q\), and one for which \(\ell < q\).
  For the first setting, where computation is performed with the logical alphabet being the same size as the physical alphabet:
  \begin{itemize}
    \item
    We generalize the results of Evans and Schulman \cite{evans2003on-the-maximum} to \(q > 2\) alphabets.
    In particular, we find that the denoising threshold with gates subject to \(q\)-ary symmetric noise of strength \(\varepsilon\) is improved over larger physical alphabet sizes (\cref{lem:lower-bound-on-majqk-denoising-threshold}, in \cref{ssec:symmetric-denoising-threshold}).
    \item
    We show that the set of functions generalizing the Boolean \textsc{xor} over larger alphabets can be reliably computed up to the previously shown denoising threshold (\cref{thm:reliable-computation-of-pa-functions}, in \cref{ssec:symmetric-computation-threshold}).
    \item
    We show that reliable universal computation is possible for \(\epsilon < (q - 1)/q\) asymptotically as \(k \rightarrow \infty\), and is possible up to the previously shown denoising threshold for finite fan-in \(k = 3\) and \(q\) prime in \cref{ssec:symmetric-computation-threshold-example}).
  \end{itemize}
  For the second setting, in which computation is performed with a subset of the available physical alphabet such that some of the physical alphabet can be employed to signal errors and help with denoising:
  \begin{itemize}
    \item
    We define error signaling protocols for \(\ell < q\) and employ these to demonstrate that denoising of two logical states over a physical ternary alphabet using \(2\)-input gates is possible for \(\varepsilon < 1 / 6\) (\cref{lem:fixed-point-convergence}, in \cref{ssec:boolean-denoising-threshold}.).
    \item
    Using error signaling we prove that universal Boolean (\(\ell = 2\)) computation using \(2\)-input gates is possible up to the conjectured denoising threshold when using \(q = 3\) (\cref{lem:reliable-boolean-computation-over-ternary-alphabet}, in \cref{ssec:boolean-computation-threshold}). 
    We use this to show that universal ternary computation up to the conjectured denoising threshold when $q=3$ is also possible (\cref{thm:universal-ternary-computation}, in \cref{sec:universal-ternary-computation}).
  \end{itemize}
  We present these results beginning with definitions and framework in \cref{sec:definitions}, followed by the large-alphabet computation \(\ell=q > 2\) scenario in \cref{sec:large-alphabet-computation} and the embedded Boolean computation \(\ell < q\) scenario in \cref{sec:universal-ternary-computation}.  
  Possible extensions and concluding observations are discussed in \cref{sec:conclusion}.

\section{Definitions and Framework}
\label{sec:definitions}
  First, we make precise the set of functions that can be computed over an alphabet of size \(q\) by a circuit using an elementary set of finitary gates.
  This set is essentially the closure of all functions that be computed by connecting our elementary gates using wires carrying \(q\)-ary signals: allowing for the permutation, duplication, and termination of wires. 
  Such a closed set is called clone in the theory of universal algebras, as we recall briefly below:
  \begin{definition}[Superpositions and Clones]
  \label{def:clone}
    Let \(F\) be a set of finitary operations \(f_i: [q]^{n_i} \rightarrow [q]\) over alphabet \([q]\).

    We call a function \(g\) a \emph{superposition over \(F\)} if it can be constructed by composing a finite number of operations in \(F\), allowing for natural wirings of variables, i.e. function inputs, between compositions.
    From \cite[Section 1.2]{lau2006function}, these rearrangements, called superposition operations, are as follows:
    \begin{enumerate}[\hspace{1cm}i)]
      \item
      permutation of variables,
      \item
      identification of variables,
      \item
      addition of new variables, and;
      \item
      substitution of variables by functions.
    \end{enumerate}

    A \emph{clone} is a set \(F\) of finitary operations \(f: [q]^n \rightarrow [q]\), \(n \ge 0\) such that
    \begin{enumerate}[\hspace{1cm}i)]
      \item 
      \(F\) contains all projection functions, i.e. \(f_i \in F\) such that \(f_i(X_1, \dots, X_n) = X_i\).
      \item
      \(F\) is closed under superposition.
    \end{enumerate}
    We refer the reader to \cite{lau2006function} for more details regarding this formalism.
  \end{definition}
  The closure of an elementary set of gates \(F\) under superposition and projections formalizes what is meant by ``the set of functions that can be computed by a circuit composed of gates from \(F\).''
  The clones of Boolean functions have been completely classified by Post \cite{post1941the-two-valued}, and classifications exist for larger alphabets \cite{lau2006function}.
  Since clones are closed under composition, we are often interested in the clone's generating set, i.e. a minimal set of gates capable of computing all functions in the clone through composition.
  At one limit, we have the complete clone:
  \begin{definition}[Universal \(q\)-ary computation]
  \label{def:universal-computation}
    We say a particular set of gates is universal for \(q\)-ary computation if the clone that it generates contains all \(n\) input functions \(f: [q]^n \rightarrow [q]\).
    This property is also commonly referred to as completeness \cite{lau2006function,slupecki1972criterion}.
  \end{definition}
  In the previous literature on Boolean computation, the \(2\)-input \nand~gate \cite{von-neumann1956probabilistic,evans1998on-the-maximum}, being itself universal, has been used both to perform computation as well as construct the denoising operation.
  This becomes more complicated over larger alphabets, as certain non-universal clones may admit larger thresholds (e.g. \cref{sec:universal-ternary-computation}).

  Additionally, in the setting of larger alphabets, a distinction can be made between the \emph{physical} and \emph{logical} alphabets.
  Specifically, operating over a physical alphabet of size \(q\), we may in general choose a smaller number of logical states over which to perform reliable computation through the use of error correction and error detection.
  We study an example where the logical alphabet is strictly smaller than the physical alphabet in \cref{sec:universal-ternary-computation}.

  Next, we formalize the notion of reliable computation in the presence of noise.
  \begin{definition}[\(\delta\)-reliability]
  \label{def:delta-reliability}
    Let \(f: [q]^n \rightarrow [q]\) be a \(q\)-ary function.
    We say a noisy circuit \(F_\varepsilon\) computes \(f\) \emph{\(\delta\)-reliably} if the probability of error over all inputs is at most \(\delta\), or equivalently
    \begin{equation}
      \min_{x_1, \dots, x_n \in [q]}\Pr[F_\varepsilon(x_1, \dots, x_n) = f(x_1, \dots, x_n)] \ge 1 - \delta.
    \end{equation}
  \end{definition}

  For large-alphabet computation (i.e. \(q > 2\)), a single parameter is no longer sufficient to characterize the noise in general;
  however, we elect to focus on gates subject to \(q\)-ary symmetric noise which is characterized by a single parameter:
  \begin{definition}[\(q\)-ary symmetric noise]
  \label{def:qary-symmetric-noise}
    We say a \(k\)-input gate over alphabet size \(q\), \(g: [q]^k \to [q]\) is \emph{\(\varepsilon\)-noisy}, denoted \(g_\varepsilon\), if for all inputs \(x_1, \dots, x_k \in [q]\), we have
    \begin{equation}
      \forall y \in [q], y \ne g(x_1, \dots, x_k) \implies \Pr[g(x_1, \dots, x_k) = y] = \frac{\varepsilon}{q - 1}.
    \end{equation}
  \end{definition}
  There are several reasons for this choice of error model.
  Firstly, it is both a natural generalization of the binary symmetric noise studied in \cite{von-neumann1956probabilistic,evans1998on-the-maximum,evans1999signal,evans2003on-the-maximum}, and the polar opposite of binary symmetric noise which concentrates all errors on a single erroneous element of the alphabet.
  In a sense, symmetric noise is also the worst case noise in the sense that, conditioned on an error having occurred, it provides no information about the correct signal.

\section{Large-alphabet computation}
\label{sec:large-alphabet-computation}
  In this Section, we analyze the threshold for reliable computation over larger alphabets i.e. \(q > 2\).
  Following the von Neumann construction, we first derive, in \cref{ssec:symmetric-denoising-threshold}, the denoising thresholds of a generalized majority gate, and then in \cref{ssec:symmetric-computation-threshold} determine the set of functions which can be reliably computed given these denoised inputs for large \(k\).
  Finally, in \cref{ssec:symmetric-computation-threshold-example}, we provide lower-bounds on the computation threshold for the special case of \(k = 3\) and \(q\) prime.

  First, let us only consider inputs subject to \(q\)-ary symmetric noise:
  \begin{definition}[\(q\)-ary symmetric \(a\)-noisy encoding]
  \label{def:a-noisy}
    We say a random variable \(X\) is an \emph{\(a\)-noisy encoding} of \(\hat{x}\) if
    \begin{equation}
    \Pr[X_i = x] =
      \begin{cases}
        1 - a, & \text{if } x = \hat{x}_i\\
        \frac{a}{q - 1}, & \text{otherwise}
      \end{cases}.
    \end{equation}
  \end{definition}

\subsection{Recovery threshold}
\label{ssec:symmetric-denoising-threshold}
  For computation with \(\varepsilon\)-noisy gates, we find that the \(k\)-input majority gate over alphabet of size \(q\) is the optimal denoising gate (in the maximum-likelihood sense) in the case where the logical alphabet size \(\ell = q\).
  This generalized majority gate, \(\maj{[q, k]}: [q]^k \to [q]\), outputs the mode of its inputs;
  ties are broken by choosing the value of the mode with the lowest input index, resulting in a deterministic balanced gate.

  Suppose that input \(X\) is subject to symmetric noise of strength \(a\), then applying a noiseless restoring gate to \(k\) independent copies, \(\{X_i\}_{i=1}^k\), of the noisy signal yields an output \(\maj{[q, k]}_0(X_1, \dots, X_k)\) with error 
  \begin{equation}
  \label{eq:restoring-output-noise-noiseless}
    m_{0}^{[q, k]}(a) 
    \equiv
    \sum_{\ell=0}^{k} \binom{k}{\ell} c^{[q, k]}_{\ell} (1 - a)^\ell a^{k - \ell}, 
  \end{equation}
  where \(c^{[q, k]}_{\ell} \in [0, 1]\) denotes the fraction of assignments over an alphabet of size \(q\) to the restoring gate taking \(k\) inputs, exactly \(\ell\) of which are correct, resulting in an incorrect output.
  Consider the limiting example of a binary alphabet: all assignments with \(\ell < \lfloor k / 2 \rfloor\) (assuming \(k\) odd) and therefore
  \begin{equation}
  \label{eq:binary-cls}
    c^{[2, k]}_{\ell} =
    \begin{cases}
      1, & \text{for } \ell < \lfloor \frac{k}{2} \rfloor \\
      0, & \text{otherwise}
    \end{cases}.
  \end{equation}
  Over larger alphabets, a strict majority is not required for the plurality operation to succeed, therefore allowing stronger denoising (i.e. \(c^{[q, k]}_{\ell} < 1\) for some \(\ell < \lfloor k / 2 \rfloor\)).

  For an \(\varepsilon\)-noisy restoring gate, we have
  \begin{equation}
  \label{eq:restoring-output-noise-noisy}
    m_{\varepsilon}^{[q, k]}(a) = \parens*{1 - \frac{\varepsilon}{q - 1}} m_{0}^{[q, k]}(a) + \varepsilon \parens*{1 - m_{0}^{[q, k]}(a)}.
  \end{equation}
  \begin{remark}[Properties of \(c^{[q, k]}_\ell\)]
  \label{rem:properties-of-c}
    Here, we take note of a few properties of \(c^{[q, k]}_\ell\).
    \begin{enumerate}
      \item
      When a majority of the inputs are correct, no assignment of the remaining inputs results in an error and therefore (assuming \(k\) odd) \(c^{[q, k]}_\ell = 0\) for all \(\ell > \lfloor k / 2 \rfloor\).
      \item
      Furthermore, if fewer than a \(1 / q\) fraction of inputs are correct, all assignments of the remaining inputs results in error, therefore \(c^{[q, k]}_\ell = 1\) for \(\ell < k / q\).
      \item
      Finally, \(c^{[q, k]}_{\ell + 1} \le c^{[q, k]}_\ell\).
    \end{enumerate}
  \end{remark}

  \begin{remark}[Uniform distribution fixed-point]
  \label{rem:maximally-mixed-fixed-point}
    For all \(\varepsilon\), the error rate has a fixed-point at \(a = (q - 1) / q\).
    \begin{align}
      m_{0}^{[q, k]}\parens*{\frac{q - 1}{q}} 
      &= \frac{1}{q^k} \sum_{\ell=0}^{\lfloor k / 2 \rfloor} \binom{k}{\ell} c^{[q, k]}_{\ell} \\
      &= \frac{q - 1}{q}.
    \end{align}
    The final equality is a consequence of the restoring gate being balanced and therefore must output each value for exactly \(1 / q\) of the \(q^k\) possible assignments of inputs.
    Examining \cref{eq:restoring-output-noise-noisy}, we see that the fixed-point of \(a = (q - 1) / q\) is preserved for all \(\varepsilon\).
  \end{remark}

  \begin{figure}
    \centering
    \includegraphics[width=0.75\textwidth]{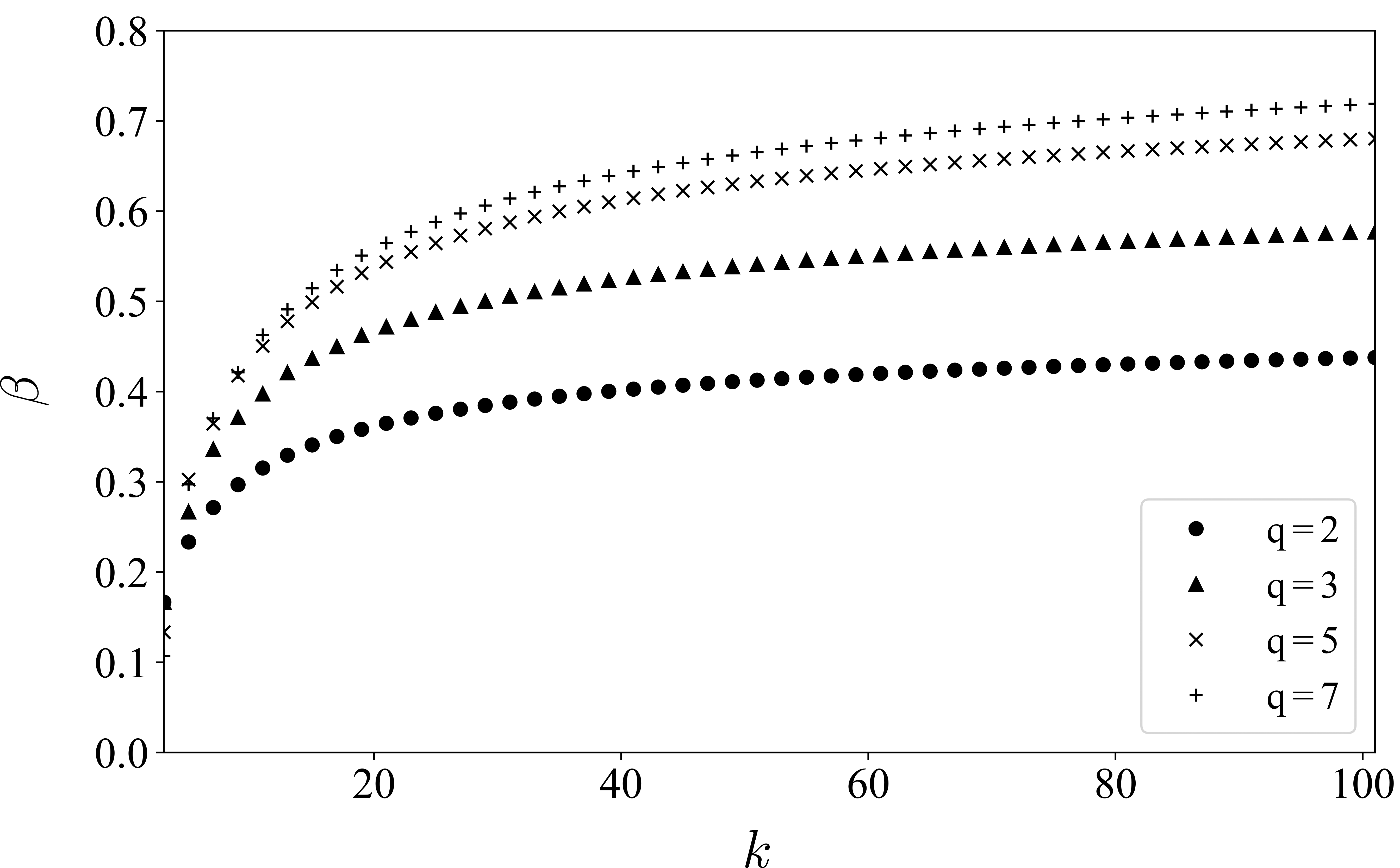}
    \caption{
    Plot of denoising threshold, \(\beta\), using the \maj{[q, k]}~gate as a function of fan-in for different alphabet sizes.
    The plotted threshold values are achieved for symmetrically noisy inputs and are known to be tight for \(q = 2\) \cite{evans2003on-the-maximum}. 
    In the large-\(k\) limit, denoising thresholds approach \((q-1)/q\).
    }
  \label{fig:denoising-threshold}
  \end{figure}

  \begin{lemma}[Lower-bound on \maj{[q, k]} denoising threshold]
  \label{lem:lower-bound-on-majqk-denoising-threshold}
    For \(q \ge 2\) and odd \(k\), let 
    \begin{equation}
    \label{eq:big-C-def}
      C^{[q, k]} \equiv 
      \frac{k}{q^{k - 1}} \sum_{\ell=0}^{\lfloor k / 2 \rfloor} \binom{k - 1}{\ell} (c^{[q, k]}_{\ell} - c^{[q, k]}_{\ell + 1}) \parens*{q - 1}^{k - \ell - 1}.
    \end{equation}

    Then if, 
    \begin{equation}
    \label{eq:maj-qk-transcritical-threshold}
        \textbf{}\varepsilon < \frac{q - 1}{q} \frac{C^{[q, k]} - 1}{C^{[q, k]}}\textbf{}
    \end{equation}
    the following hold:
    \begin{enumerate}[\hspace{1cm}i)]
      \item
      \(\exists \mu_\varepsilon^{[q, k]} \in [0, \frac{q - 1}{q})\) such that \(\forall a \in [0, \frac{q-1}{q}), \lim_{n \rightarrow \infty} m^{[q, k] n}_\varepsilon(a) = \mu_\varepsilon^{[q, k]}\); and
      \item
      \(\exists \nu_\varepsilon^{[q, k]} \in (\frac{q - 1}{q}, 1]\) such that \(\forall a \in (\frac{q-1}{q}, 1], \lim_{n \rightarrow \infty} m^{[q, k] n}_\varepsilon(a) = \nu_\varepsilon^{[q, k]}\).
    \end{enumerate}
  \end{lemma}
  \begin{proof}
    The function \(m^{[q, k]}_\varepsilon(a)\) has at most three fixed-points.
    We show that for all \(\varepsilon < \frac{q - 1}{q} \frac{C^{[q, k]} - 1}{C^{[q, k]}}\), the fixed-point at \(a = (q-1)/q\) given by \cref{rem:maximally-mixed-fixed-point} is unstable.

    Taking the derivative, 
    \begin{align}
      \frac{d m_{0}^{[q, k]}}{da}
      &=
      \frac{d}{da} \sum_{\ell=0}^{\lfloor k / 2 \rfloor} \binom{k}{\ell} c^{[q, k]}_{\ell} (1 - a)^\ell a^{k - \ell} \\
      &=
      \sum_{\ell=0}^{\lfloor k / 2 \rfloor} (k - \ell) \binom{k}{\ell} c^{[q, k]}_{\ell} (1 - a)^\ell a^{k - \ell - 1} - \sum_{\ell=0}^{\lfloor k / 2 \rfloor} \ell \binom{k}{\ell} c^{[q, k]}_{\ell} (1 - a)^{\ell - 1} a^{k - \ell} \\
      &=
      k \sum_{\ell=0}^{\lfloor k / 2 \rfloor} \binom{k - 1}{\ell} c^{[q, k]}_{\ell} (1 - a)^\ell a^{k - \ell - 1} - k \sum_{\ell=0}^{\lfloor k / 2 \rfloor - 1} \binom{k - 1}{\ell} c^{[q, k]}_{\ell + 1} (1 - a)^{\ell} a^{k - \ell - 1} \\
      &=
      k \sum_{\ell=0}^{\lfloor k / 2 \rfloor} \binom{k - 1}{\ell} (c^{[q, k]}_{\ell} - c^{[q, k]}_{\ell + 1}) (1 - a)^\ell a^{k - \ell - 1} \,.
      \label{eq:m-first-derivative}
    \end{align}

    Taking another derivative,
    \begin{align}
      \frac{d^2 m_{0}^{[q, k]}}{da^2}
      &=
      k \frac{d}{da} \sum_{\ell=0}^{\lfloor k / 2 \rfloor} \binom{k - 1}{\ell} (c^{[q, k]}_{\ell} - c^{[q, k]}_{\ell + 1}) (1 - a)^\ell a^{k - \ell - 1}  \\
      &=
      k (k - 1) \sum_{\ell=0}^{\lfloor k / 2 \rfloor} \binom{k - 2}{\ell} (c^{[q, k]}_{\ell} - 2 c^{[q, k]}_{\ell + 1} + c^{[q, k]}_{\ell + 2}) (1 - a)^\ell a^{k - \ell - 2}.
    \end{align}

    As a result, \(m^{[q, k]}_{\varepsilon}(a) = a\) has three distinct solutions so long as 
    \begin{equation}
      \left. \frac{d m_{\varepsilon}^{[q, k]}}{da} \right|_{a = (q - 1) / q} = 1.
    \end{equation}

    From \cref{eq:m-first-derivative},
    \begin{align}
      \left. \frac{d m_{0}^{[q, k]}}{da} \right|_{a = (q-1)/q}
      &=
      \frac{k}{q^{k - 1}} \sum_{\ell=0}^{\lfloor k / 2 \rfloor} \binom{k - 1}{\ell} (c^{[q, k]}_{\ell} - c^{[q, k]}_{\ell + 1}) \parens*{q - 1}^{k - \ell - 1} \\
      &= C^{[q, k]}.
    \end{align}
    Therefore,
    \begin{align}
      \left. \frac{d m_{\varepsilon}^{[q, k]}}{da} \right|_{a = (q-1)/q}
      &= 
      \parens*{1 - \frac{\varepsilon}{q - 1}} \left. \frac{d m_{0}^{[q, k]}}{da} \right|_{a = (q-1)/q} - \varepsilon \left. \frac{d m_{0}^{[q, k]}}{da} \right|_{a = (q-1)/q},
    \end{align}
    which shows that
    \begin{equation}
      \varepsilon < \frac{q - 1}{q} \frac{C^{[q, k]} - 1}{C^{[q, k]}} \implies \left. \frac{d m_{0}^{[q, k]}}{da} \right|_{a = (q-1)/q} > 1,
    \end{equation}
    and thus the Lemma holds.
  \end{proof}
  By setting \(q = 2\) in \cref{lem:lower-bound-on-majqk-denoising-threshold} and using \cref{eq:binary-cls}, we recover the Boolean denoising result of \cite{evans2003on-the-maximum}.
  Furthermore, in the case of computation over the Boolean alphabet, this bound is tight;
  we note however that this is not necessarily true for \(q > 2\), in which case stable fixed-points exist above the threshold of \cref{lem:lower-bound-on-majqk-denoising-threshold}, but denoising is no longer possible for the full range \(a \in [0, (q-1)/q)\) as the fixed-point corresponding to the \(((q-1)/q)\)-noisy, uniform distribution, fixed-point becomes stable.
  For brevity, we will refer to the threshold of \cref{eq:maj-qk-transcritical-threshold} as the denoising threshold, or more specifically, the point of transcritical bifurcation, as above this point a non-measure zero subset of the simplex is drawn towards the uniform distribution upon repeated denoising.
  The point of transcritical bifurcation can be contrasted with the ultimate saddle-node bifurcation after which point only one stable fixed-point remains (see \cref{fig:fixed-point-diagram}).
  The positive results for various alphabet sizes \(q\) and fan-ins \(k\) are shown in \cref{fig:denoising-threshold}.

  \begin{figure}
    \centering
    \includegraphics[width=0.65\textwidth]{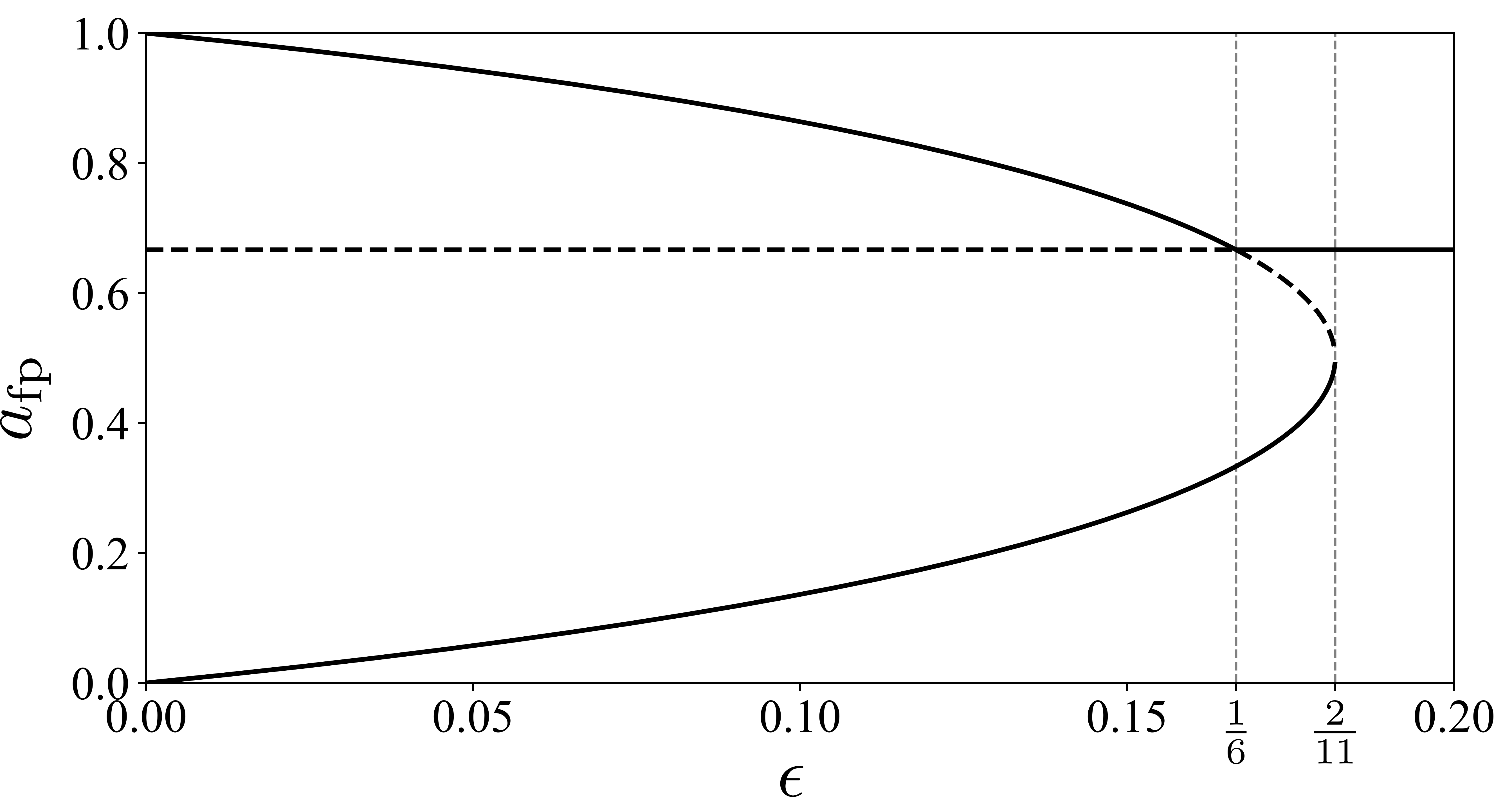}
    \caption{
  Diagram showing stable (solid lines) and unstable (dashed lines) fixed-points of the \(\varepsilon\)-noisy \maj{[3, 3]} gate for symmetrically \(a\)-noisy inputs.
  Note that at the denoising lower-bound of \(\varepsilon = 1 / 6\) (\cref{lem:lower-bound-on-majqk-denoising-threshold}) corresponds to a transcritical bifurcation, at which point the qualitative structure of the fixed-points changes, with the central \((2/3)\)-noisy fixed-point becoming stable.
  Two stable fixed-points persist until the ultimate saddle-node bifurcation at \(\varepsilon = 2 / 11\), corresponding to a discontinuous phase transition.
  }
  \label{fig:fixed-point-diagram}
  \end{figure}

  We now generalize our denoising result to all input distributions (i.e. beyond \(q\)-ary symmetric noise).
  Our approach is to analyze denoising as a discrete-time dynamical system.

  First, we specify a categorical distribution over a \(q\)-ary alphabet as a point in the simplex
  \begin{equation}
    \Delta_q = \left\{\vec{p} \in \mathbb{R}^q \middle| p_i \ge 0, \sum_{i=1}^q p_i = 1 \right\} \subset \mathbb{R}^q.
  \end{equation}
  For independent and identically distributed (i.i.d.) \(q\)-ary input random variables \(X\) with distribution parameterized by vector \(\vec{p} \in \Delta_q\), denote its output \(Y_\varepsilon = \maj{[q, k]}(X_1, \dots, X_k)\) with distribution parameterized by \(\vec{q} \in \Delta_q\).
  Denote the map from input distributions to output distributions \(\mathcal{M}^{[q, k]}: \Delta^q \rightarrow \Delta^q\) such that \(\mathcal{M}(\vec{p}) = \vec{q}\), where we have omitted superscript \([q, k]\) for brevity.

  In this picture, the fixed-points of \cref{lem:lower-bound-on-majqk-denoising-threshold} correspond to points \(\vec\chi^{(i)} \in \Delta_q\) such that \(\vec\chi^{(i)}_i = 1 - \nu_\varepsilon\) and \(\vec\chi^{(i)}_j = \nu_\varepsilon / (q - 1)\) for \(j \ne i\), where again superscripts of \([q, k]\) and subscripts of \(\varepsilon\) have been omitted.

  Further, let us denote by \(\mathcal{R}^{(i)} \subset \Delta_q\), the regions in the probability simplex decoding to logical character \(i\).
  Denoising using the \maj{[q,k]} gate, we will show that it makes sense to define the regions as follows:
  \begin{equation}
    \mathcal{R}^{(i)} = \left\{\vec{p} \in \Delta_q ~\middle|~ \forall j \in [q] \setminus \{i\},~ p_i - p_j > 0\right\}.
  \end{equation}
  That is, \(\mathcal{R}^{(i)}\) corresponds to the region in the probability simplex over \(q\) elements, \(\Delta_q\), where symbol \(i \in [q]\) is the most likely character.

  Our goal is to show that below the denoising threshold, all points in \(\mathcal{R}_0^{(i)}\) flow towards their respective fixed-points \(\vec\chi^{(i)}\) under repeated iteration of the map \(\mathcal{M}\).

  \begin{lemma}[All points in \(\mathcal{R}^{(i)}\) approach \(\vec{\chi}^{(i)}\) upon repeated iteration of \(\mathcal{M}\)]
  \label{lem:symmetric-denoising-basin}
    For \(q \ge 2\), \(k \ge 3\), and \(\varepsilon\) below threshold, all points \(\vec{p} \in \mathcal{R}^{(i)}\) satisfy
    \begin{equation}
      \lim_{n \rightarrow \infty} \mathcal{M}^n(\vec{p}) = \vec{\chi}^{(i)}.
    \end{equation}
  \end{lemma}
  \begin{proof}
    Without loss of generality, consider a \(q\)-ary random variable with distribution \(\vec{p} \in \mathcal{R}^{(0)}_0\) with elements in sorted order (i.e. \(i > j \implies p_i \ge p_j\)).
    After one application of the map \(\mathcal{M}\), we have \(\vec{p'} = \mathcal{M}(\vec{p})\).

    Note that \(p'_0\) is a convex combination of convex functions of the form 
    \begin{equation*}
      p_0^{k_0} p_1^{k_1} \dots p_{q-1}^{k_{q-1}},
    \end{equation*}
    and is therefore itself jointly convex in \(p_1, \dots, p_q\) and for fixed \(p_0\) is maximized at the boundary \(p_i = (1 - p_0) / (q - 1)\) for all \(i > 0\).
    Intuitively, for fixed \(p_0\), the minimum chance of confusion, and therefore maximum amplification, is achieved when all erroneous inputs are equally likely.
    As a result, the probability of error after denoising \(\alpha' \equiv 1 - p_0'\) is minimized for the case of symmetric noise studied in \cref{lem:lower-bound-on-majqk-denoising-threshold},
    \begin{equation}
      \label{eq:error-lower-bound}
      \alpha' \ge m^{[q, k]}_\varepsilon(1 - p_0).
    \end{equation}

    We can devise a complementary bound by noticing that, for fixed gap \(\delta_1 \equiv p_0 - p_1\), the gap \(\delta_1\) is also convex in the region \(\mathcal{R}^{(0)}\) since \(p_0 > p_1\).
    Here, note that \(\delta_1'\) is minimized for symmetrically noisy inputs \(p_i = (1 - p_0) / (q - 1)\) for all \(i > 0\).
    As a result, the gap after denoising is minimized for
    \begin{align}
      \delta_1' = p_0' - p_1' &\ge (1 - m^{[q, k]}_\varepsilon(1 - p_0)) - \frac{m^{[q, k]}_\varepsilon(1 - p_0)}{q - 1},\\
      \label{eq:error-upper-bound}
      \implies \alpha' = 1 - p_0' &\le m^{[q, k]}_\varepsilon(1 - p_0) - \parens*{p_1' - \frac{m^{[q, k]}_\varepsilon(1 - p_0)}{q - 1}}.
    \end{align}
    Together, the bounds \cref{eq:error-upper-bound} and \cref{eq:error-lower-bound} along with the result for symmetrically noisy inputs of \cref{lem:lower-bound-on-majqk-denoising-threshold}, gives the desired result.
  \end{proof}

  \begin{figure}
    \centering
    \includegraphics[width=0.99\textwidth]{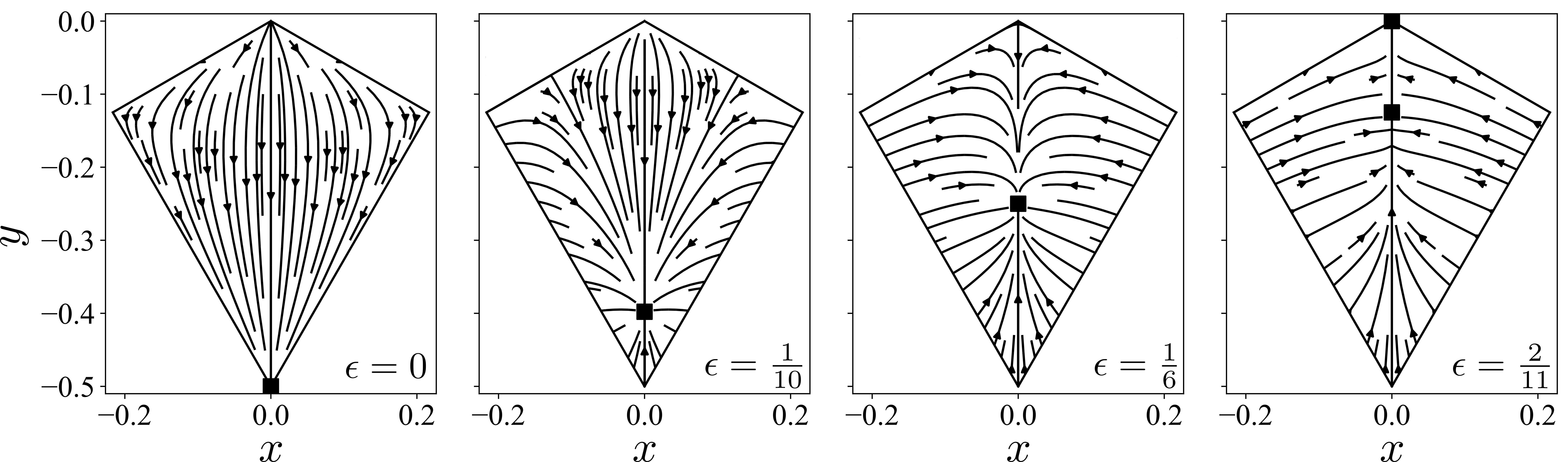}
    \caption{
    Streamlines of the vector field \(\bar{\mathcal{M}}_\varepsilon(x, y) - (x, y)\) for four values of \(\varepsilon\) in the region \(\bar{\mathcal{R}}^{(2)}\) along with stable fixed-points (square markers).
    The picture in regions \(\bar{\mathcal{R}}^{(0)}\) and \(\bar{\mathcal{R}}^{(1)}\) can be obtained by symmetry. 
    Note that the fixed-point at the transcritical bifurcation (\cref{lem:lower-bound-on-majqk-denoising-threshold}) \(\varepsilon = 1/6\) is \((1/3)\)-noisy (third plot).
    Between the transcritical bifurcation and saddle-node bifurcation (i.e. \(1 / 6 < \varepsilon < 2 / 11\)), denoising is still possible for a strict subset of the original region \(\bar{\mathcal{R}}^{(2)}\) (fourth plot) owing to the emergence of a fourth stable fixed-point at the center of the simplex.
    }
  \label{fig:denoising-stream-plot-symmetric}
  \end{figure}

  For visualization purposes, and for consistency with \cref{sec:universal-ternary-computation}, we introduce a transformed set of coordinates for \(\Delta_3\), defined by
  \begin{equation}
  \label{eq:xy-reparameterization}
    x = \frac{\sqrt{3}}{4}(p_1 - p_0), 
    \qquad\text{and}\qquad
    y = \frac{3}{4}\left(p_0 + p_1 - \frac{2}{3}\right).
  \end{equation}
  We will denote the action of the transformation using barred variables \(\bar{\mathcal{M}}_\varepsilon\), and the corresponding logical regions \(\bar{\mathcal{R}}^{(i)}\).
  Note that, in these coordinates, the origin \(x = y = 0\) corresponds to the uniform probability distribution \((1/3, 1/3, 1/3)\), and the level curves of the coordinate \(p_2\) are horizontal lines. 
  A visualization of the denoising dynamics is shown for \(q = k = 3\) in \cref{fig:denoising-stream-plot-symmetric}.

\subsection{Computation threshold}
\label{ssec:symmetric-computation-threshold}
  First, we consider the class of operations that may be performed using only gates that preserve the symmetric noise of \cref{def:qary-symmetric-noise} and show that such gates have the property that they can be computed up to the ultimate threshold of \((q - 1) / q\) for all \(k \ge 2\).
  While this class of operations is not itself universal, it may be augmented to a universal set using gates which can be computed reliably up to threshold \((q - 1) / q\) asymptotically as \(k \rightarrow \infty\).

  \begin{definition}[Symmetric noise preserving gates]
  \label{def:symmetric-noise-preserving-gates}
    Let \(\{X_i\}_{i=1}^k\) be independent, but not necessarily identically distributed, \(a_i\)-noisy encodings of \(\{\hat{x}_i\}_{i=q}^k\).

    A gate \(g: [q]^k \to [q]\) is said to be \emph{symmetric noise preserving (SNP)}, if the output
    \begin{equation}
      Y = g(X_1, \dots, X_n),
    \end{equation}
    is \(a\)-noisy where \(a\) is a function only of \(\{a_i\}_{i=1}^k\).

    Note that since the matrix for \(q\)-ary symmetric channel of \cref{def:a-noisy} is invertible for \(\varepsilon < (q - 1) / q\), this property holds independent of \(\varepsilon\) and therefore nothing is lost by considering noiseless gates.
  \end{definition}

  \begin{lemma}[SNP gates are not universal for \(q > 2\)]
  \label{lem:non-universality-of-snp}
    The clone of symmetric noise preserving functions is not universal for alphabet of size \(q > 2\).
  \end{lemma}
  \begin{proof}
    Let \(g\) be a \(k\)-input SNP gate.
    Then for all indices \(1 \le i \le k\), and all choices \(\vec{c} = (c_j)_{j \ne i}\), the restricted function \(g_{\vec{c}, i}: [q] \rightarrow [q]\),
    \begin{equation}
    \label{eq:snp-condition}
      g_{\vec{c}, i}(x) \equiv g(c_1, \dots, c_{i-1}, x, c_{i+1}, c_k),
    \end{equation}
    must either be a constant function or a bijection.

    To see this, suppose there is a function \(g\) which does not satisfy \cref{eq:snp-condition} for some \(i\) and \(\vec{c}\).
    Consider inputs \(X_j\) which are \(0\)-noisy encodings of \(\hat{x}_j = c_j\) for all \(j \ne i\);
    and \(X_i\) an \((a_i > 0)\)-noisy encoding of \(\hat{x}_i\).
    Further let \(\hat{g} = g_{\vec{c}, i}(\hat{x}_i)\).
    Since \(g_{\vec{c}, i}\) is neither constant nor a bijection, for \(q > 2\) there are \(y, y' \in [q] \setminus \hat{g}\) such that \(\Pr[G = y] = 0\) and \(\Pr[G = y'] > 0\).
    Therefore \(G = g(X_1, \dots, X_k)\) is not a \(a\)-noisy.

    Furthermore, the set of SNP functions is closed under composition and therefore forms a clone.
    Therefore any functions computed using only SNP gates must satisfy \cref{eq:snp-condition}, and therefore SNP gates are not universal for \(q\)-ary computation.
  \end{proof}
  If we place additional constraints on the inputs, there may exist larger sets of gates that preserve symmetric noise.
  The denoising gate of \cref{ssec:symmetric-denoising-threshold}, is a notable example which does not satisfy the conditions of \cref{lem:pseudo-additive} but nevertheless preserves symmetric noise for identically distributed inputs.

  We now consider the following class of functions that generalize the Boolean \textsc{xor} over larger alphabets:
  \begin{definition}[Pseudo-additive functions]
  \label{def:pseudo-additive}
    A \(k\)-input \(q\)-ary function \(g\) is called \emph{pseudo-additive (PA)} if it can be written in the following form:
    \begin{equation}
    \label{eq:pseudo-additive}
      g(x_1, \dots, x_k) = \sum_{i=1}^k \sigma_i(x_i) \pmod{q},
    \end{equation}
    where each \(\sigma_{i}: [q] \to [q]\) is either constant function or a bijection.
  \end{definition}
  
  Note that the set of pseudo-additive functions is closed under composition and therefore forms a clone.
  Additionally, we show that they preserve symmetric noise.
  \begin{lemma}[\(\text{PA} \subset \text{SNP}\)]
  \label{lem:pseudo-additive}
    A \(k\)-input \(q\)-ary pseudo-additive function is symmetric noise preserving.
  \end{lemma}
  \begin{proof}
    First, we show that this condition is sufficient.
    Let \(G_r \equiv \sum_{i=1}^r \sigma_i(X_i) \pmod{q}\) and \(\hat{g}_r \equiv \sum_{i=1}^r \sigma_i(\hat{x}_i) \pmod{q}\).
    First we show that this condition is sufficient for unary gates \(g\).
    If the function is constant, then \(\Pr\bracks*{g(X_1) = y} = \delta_{y, \hat{g}_1}\) and is trivially \(b\)-noisy with \(b = 0\).
    If the function is a surjection, then
    \begin{equation}
      \Pr\bracks*{g(X_1) = y} = 
      \begin{cases}
        1 - a_1, & \text{if } y = \hat{g}_1 \\
        \frac{a_1}{q-1}, & \text{otherwise}
      \end{cases},
    \end{equation}
    and is \(b\)-noisy with \(b = a_1\).

    We proceed by induction on \(k\).
    Suppose the Lemma holds for \(k\)-input functions giving a \(b\)-noisy output.
    \begin{align}
      \begin{split}
      \Pr[&G_{k+1} = y] = \Pr[G_k + \sigma_{k+1}(X_{k+1}) = y \pmod{q}] = \\
      & \Pr[G_k + \sigma_{k+1}(X_{k+1}) = y \pmod{q} | G_{k} = \hat{g}_k, X_{k+1} = \hat{x}_{k+1}] \Pr[G_{k} = \hat{g}_k] \Pr[X_{k+1} = \hat{x}_{k+1}] + \\
      &\Pr[G_k + \sigma_{k+1}(X_{k+1}) = y \pmod{q} | G_{k} \ne \hat{g}_k, X_{k+1} = \hat{x}_{k+1}] \Pr[G_{k} \ne \hat{g}_k] \Pr[X_{k+1} = \hat{x}_{k+1}] + \\
      &\Pr[G_k + \sigma_{k+1}(X_{k+1}) = y \pmod{q} | G_{k} = \hat{g}_k, X_{k+1} \ne \hat{x}_{k+1}] \Pr[G_{k} = \hat{g}_k] \Pr[X_{k+1} \ne \hat{x}_{k+1}] + \\
      &\Pr[G_k + \sigma_{k+1}(X_{k+1}) = y \pmod{q} | G_{k} \ne \hat{g}_k, X_{k+1} \ne \hat{x}_{k+1}] \Pr[G_{k} \ne \hat{g}_k] \Pr[X_{k+1} \ne \hat{x}_{k+1}]
      \end{split}\\
      \begin{split}
        &=\delta_{y \hat{g}_{k+1}} (1 - a_{k+1})(1 - b) + \frac{1 - \delta_{y \hat{g}_{k+1}}}{q - 1} (a_{k+1} + b - 2ab) + \frac{q - 2 + \delta_{y \hat{g}_{k+1}}}{(q - 1)^2} a_{k+1} b,
      \end{split}
    \end{align}
    where the second equality is an application of the chain rule of probability; and the third equality makes use of assumptions that \(G_{k}\) and \(\sigma_{k+1}(X_{k+1})\) are \(b\)-noisy and \(a_{k+1}\)-noisy (or \(0\)-noisy if \(\sigma_{k+1}\) is constant) respectively, as well as the fact that modular addition randomizes equally among all elements of the alphabet.

    We then have
    \begin{equation}
      \Pr[G_{k+1} = y] = 
      \begin{cases}
        (1 - a_{k+1})(1 - b) + \frac{a_{k+1} b}{q - 1}, & \text{if } y = \hat{g}_{k+1}\\
        \frac{a_{k+1} + b - 2a_{k+1} b}{q - 1} + \frac{(q - 2) a_{k+1} b}{(q - 1)^2} & \text{otherwise}
      \end{cases},
    \end{equation}
    and is therefore symmetric noise preserving i.e. \(\text{PA} \subseteq \text{SNP}\).

    Comparing with \cref{def:symmetric-noise-preserving-gates}, we see that for \(q > 3\), the containment is strict.
  \end{proof}

  \begin{theorem}
  \label{thm:reliable-computation-of-pa-functions}
    The clone of pseudo-additive functions can be reliably computed using \(\varepsilon\)-noisy single-input permutation gates, the single-input constant function, \(2\)-input modular addition gates, and majority gates \(\maj{[q,k]}\)~gates up to the majority gate denoising threshold of \cref{lem:lower-bound-on-majqk-denoising-threshold}.
  \end{theorem}
  \begin{proof}
    First, we show that \(\varepsilon\)-noisy pseudo-additive functions \(g\), given \(a\) noisy encodings, maintain their output in the correct basin of attraction for \(\varepsilon < (q - 1)/q\).

    Let \(\sigma_\varepsilon: [q] \rightarrow [q]\) denote an \(\varepsilon\)-noisy permutation gate or constant gate and \(X\) be an \(a\)-noisy encoding of \(\hat{x}\).
    If \(\sigma\) is a constant gate, then its output is always \(\varepsilon\)-noisy and therefore the output can be denoised up to \(\varepsilon < (q - 1) / q\).
    Therefore we focus on the case of a permutation \(\sigma\), in which case \(\sigma_\varepsilon(X)\) is a \(b\)-noisy encoding of \(\sigma_0(\hat{x})\), where
    \begin{equation}
      b = \varepsilon + \bracks*{1 - \parens*{\frac{q}{q - 1}} \varepsilon} a.
    \end{equation}
    Therefore for all \(q \ge 2\),
    \begin{equation}
      a < \frac{q - 1}{q} \implies b < \frac{q - 1}{q}.
    \end{equation}
    Therefore the permutation gate can reliably compute up to the denoising threshold.

    Next, consider the \(\varepsilon\)-noisy modular addition gate \(\textsc{add}_\varepsilon^{(q)}(x_1, x_2) = x_1 + x_2 \mod{q}\).
    If its inputs \(X_i\) are \(a_i\)-noisy encodings of \(\hat{x}_i\) for \(i \in \{1, 2\}\), then \(\textsc{add}_\varepsilon^{(q)}(X_1, X_2)\) is \(b'\)-noisy with
    \begin{equation}
      b' = \varepsilon + \frac{\parens*{a_1 (q - 1) + \parens*{1 - a_1} a_2 q - a_2} (q (1 - \varepsilon) - 1)}{(q - 1)^2}.
    \end{equation}
    Therefore for all \(q \ge 2\),
    \begin{equation}
      a_1 < \frac{q - 1}{q}, \text{ and } a_2 < \frac{q - 1}{q} \implies b' < \frac{q - 1}{q}.
    \end{equation}
    This shows that the output under either \(\{\sigma, \textsc{add}^{q}\}\) always remains in the correct basin of attraction for \(\maj{[q,k]}\) if its inputs are sufficiently close to an \(a\)-noisy encodings for some \(a < (q-1) / q\).
    This output can then be reliably denoised through repeated application of \(\maj{[q,k]}\).
    As the set of gates \(\{\sigma, \textsc{add}^{q}\}\) generates the clone of PA functions, the clone of PA functions can be reliably computed for \(\varepsilon < (q - 1) / q\) for all \(k\).
  \end{proof}

  While the set of \(\text{SNP}\) gates are not themselves universal (\cref{lem:non-universality-of-snp}), they can be augmented into a universal set using a general \(2\)-input operation over alphabet \(q\) \cite{lau2006function}.
  All such operations can be reliably computed up to the ultimate threshold of \((q - 1) / q\) as \(k \rightarrow \infty\); for example, by using redundant inputs to perform native error correction (in the spirit of the \xnand~for \(q = 2\)).
  \begin{remark}[Reliable universal computation in the \(k \rightarrow \infty\) limit]
  \label{rem:reliable-computation-limit}
    Suppose that \(g: [q] \times [q] \rightarrow [q]\) is a \(2\)-input gate, and suppose we are able to prepare signals \(X\) and \(Y\) which are sufficiently close to \(a\)-noisy with \(a < (q - 1) / q\).
    Further suppose that all gates are \(\varepsilon\)-noisy for \(\varepsilon < (q - 1) / q\).
    We remark that there exists an integer \(K\), such that there exists a gate \(g^*\) with fan-in \(2K\) which can be used to compute \(g\) on inputs \(X_1, \dots, X_K\) and \(Y_1, \dots, Y_K\), maintaining the output in the correct part of the probability simplex for all inputs, where \(X_i\) and \(Y_i\) are i.i.d. versions of \(X\) and \(Y\) respectively.
    The idea is for \(g^*\) to be a function that natively performs error correction on its inputs, i.e. \(g^*\) computes \(g(x, y)\) where \(x = \maj{[q,k]}(X_1, \dots, X_K)\) and \(y= \maj{[q,k]}(Y_1, \dots, Y_K)\).

    \noindent
    Then, its output
    \begin{equation}
      Z^{(K)} = g^*(x_1, \dots, x_K, y_1, \dots, y_k)
    \end{equation}
    errs with a probability that scales as \(e^{-\Theta(K)}\) and we can always choose a \(K\) such that the output of the \(\varepsilon\)-noisy \(g^*_\epsilon\) is always within the correct part of the probability simplex.
  \end{remark}
  
  Thus for any \(\varepsilon < (q - 1) / q\), there exists a sufficiently large fan-in \(K^*\) for which reliable universal computation is possible, i.e. choose \(K^*\) such that \(\varepsilon\) is below the denoising threshold \cref{lem:lower-bound-on-majqk-denoising-threshold} and below the threshold required by \cref{rem:reliable-computation-limit} for a set of \(2\)-input gates which augment the set of PA functions into a universal set.

\subsection{Example \(k = 3\) and \(q\) prime}
\label{ssec:symmetric-computation-threshold-example}
  In addition to the asymptotic result argued in \cref{ssec:symmetric-computation-threshold}, we show that the denoising threshold may be achieved for finite \(k\).
  In fact, we demonstrate that for \(k = 3\) and \(q\) prime, reliable universal computation can be achieved up to the denoising lower-bound of \cref{lem:lower-bound-on-majqk-denoising-threshold}.
  \begin{lemma}
  \label{lem:symmetric-computation-threshold-example}
    Reliable universal computation over an alphabet of size \(q\), where \(q\) is prime, can be performed using \(\varepsilon\)-noisy \(2\)-input pseudo-additive gates, \(2\)-input modular multiplication gates, and majority gates \(\maj{[q,3]}\)~gates up to the majority gate denoising threshold of \(\varepsilon < (q - 1)/(q (q + 1))\) given by \cref{lem:lower-bound-on-majqk-denoising-threshold}.
  \end{lemma}
  \begin{proof}
    First, we calculate the denoising threshold, i.e. the point of transcritical bifurcation for the \(\maj{[q, 3]}\)~gate.
    Following the result in \cref{lem:lower-bound-on-majqk-denoising-threshold},
    \begin{equation}
    \label{eq:small-c-k=3}
      c^{[q, 3]}_0 = 1, \qquad
      c^{[q, 3]}_1 = 1 - \frac{q - 2}{3 (q - 1)}, \qquad
      c^{[q, 3]}_2 = 0, \qquad\text{and}\qquad
      c^{[q, 3]}_3 = 0.
    \end{equation}
    This gives a denoising threshold lower-bound of
    \begin{equation}
      C^{[q, 3]} = \frac{q + 1}{q} \implies \beta^{[q, 3]} \ge \frac{q - 1}{q} \frac{C^{[q, 3]} - 1}{C^{[q, 3]}} = \frac{q - 1}{q (q+1)}.
    \end{equation}
    Substituting the coefficients from \cref{eq:small-c-k=3} into \cref{eq:restoring-output-noise-noisy}, we find that at the transcritical bifurcation, the outputs are \((1/q)\)-noisy.
    As an aside, a similar calculation shows that the saddle-node bifurcation occurs for \(k = 3\) at 
    \begin{equation} 
      \varepsilon = \frac{q - 1}{5q - 4},
    \end{equation}
    with \((1/2)\)-noisy outputs, which can be seen as the ultimate denoising threshold.

    From \cref{thm:reliable-computation-of-pa-functions}, we have that pseudo-additive functions can be calculated up to the denoising threshold.
    To obtain a lower-bound for reliable universal computation, consider augmenting the pseudo-additive gates with the \(2\)-input modular multiplication operation:
    \begin{equation}
      \textsc{mul}^{q}(X_1, X_2) = X_1 \times X_2 \pmod{q}.
    \end{equation}

    To see this, note that the most likely error is that the output of two non-zero elements is mistaken for `\(0\).'
    For symmetrically noisy \(a\)-noisy inputs and prime \(q\),
    \begin{align}
      p_{\text{correct}} &= (1 - a)^2 + \frac{q - 2}{(q-1)^2} a^2, \\ 
      p_{0} &= \parens*{\frac{a}{q - 1}}^2 + 2 \parens*{\frac{a}{q - 1}} \parens*{1 - \frac{a}{q - 1}},
    \end{align}
    where, in calculating these probabilities, we have used the fact that all elements in \([q]\) have inverses for \(q\) prime.
    We find that \(p_{\text{correct}} > p_0\) requires \(a < 1 - 1 / \sqrt{q}\).

    Taken together with the error rate of the denoising operation from above, we find that for prime \(q\) and \(k = 3\), reliable universal computation is achievable up to the denoising threshold of \(\varepsilon < (q - 1)/(q (q + 1))\).
  \end{proof}
  Notably, this is a consequence of the fact that the distribution does not tend to the uniform distribution fixed-point at the point of the transcritical bifurcation, as in \cref{fig:fixed-point-diagram} for the \(q = 3\) case.
  This is in stark contrast with the previously studied Boolean case (\(q = 2\)).
  Practically, this offers some advantages: instead of requiring a uniquely designed gate such as the \xnand~to achieve the optimal computation threshold, for \(k = 3\) and \(q\) prime, the standard modular multiplication operation can be reliably computed up to the transcritical bifurcation of \cref{eq:maj-qk-transcritical-threshold}.

\section{Universal Ternary computation using error signaling}
\label{sec:universal-ternary-computation}
  In this Section, we design a set of gates for performing universal ternary computation, showing that we can improve on the fault-tolerance threshold given in \cite{evans1998on-the-maximum} if we add a third alphabet character to perform error signaling. 
  In \cref{ssec:projecting-gates}, we construct projecting gates that allow us to rewrite any ternary computation problem as a binary computation problem. Then, in \cref{ssec:boolean-denoising-threshold} and \cref{ssec:boolean-computation-threshold}, we construct a denoising gate and a computation gate that are sufficient for universal Boolean computation over the logical alphabet. In \cref{ssec:lifting-gates}, we construct a lifting gate that allows us to lift the result of the computation to a ternary output. Finally, in \cref{ssec:proving_ternary_computation}, we prove the following theorem, which is the main result of this Section. 
  \begin{theorem}
    \label{thm:universal-ternary-computation}
    For \((k = 2)\)-input gates over an alphabet of size $q = 3$, reliable universal computation over the ternary alphabet is possible using the following $\varepsilon$-noisy gates up to the denoising threshold of $\varepsilon < 1/6$: 
    \begin{itemize}
      \item projection gates $\projone_\varepsilon,\projtwo_\varepsilon$,
      \item denoising gates $\den_\varepsilon$,
      \item computation gates $\errornand_\varepsilon$,
      \item lifting gates $\lift_\varepsilon$.
    \end{itemize}
  \end{theorem}
\noindent Throughout this Section, we assume that \(\varepsilon < 1/6 \).

  \subsection{Projecting gates}
  \label{ssec:projecting-gates}
  Consider projecting gates \(\projone: [3]^2 \to [3], \projtwo: [3]^2 \to [3]\) with truth tables given in \cref{tab:ternary-q3-k2-projectors}. For each ternary input $x$ to the computation that we are interested in, we will duplicate $x$ and compute \(\projone_\varepsilon(x,x)\) and $\projtwo_\varepsilon(x,x)$ (the $\varepsilon$-noisy versions of $\projone,\projtwo$). It remains to be checked that $\projone_\varepsilon$ is more likely to output the least significant bit of the binary representation of $x$ than its complement, while $\projtwo_\varepsilon$ is more likely to output the most significant bit of the binary representation of $x$ than its complement. 
  This is necessary and sufficient because we will be able to denoise both outputs using the denoising gate if and only if each output is more likely to be correct than not, as we shall see in \cref{lem:fixed-point-convergence}. Noting that applying symmetric noise does not change the ordering of which physical alphabet elements are most likely to be outputted, it suffices to consider the versions of the gates without noise, and it is easy to see that for each possible input $x \in [3]$, the outputs of $\projone(x,x)$ and $\projtwo(x,x)$ return the least and most significant bits of the binary representation of $x$ respectively.

  \begin{table}[htb]
    \centering
    \setlength{\tabcolsep}{5pt}
    \begin{tabular}{|c|c|c|c|}
    \hline
    $x$ & $y$ & \(\projone(x,y)\) & \(\projtwo(x,y)\) \\
    \hline
    0 & 0 & 0 & 0\\
    \hline
    0 & 1 & 2 & 0\\
    \hline
    0 & 2 & 0 & 2\\
    \hline
    1 & 1 & 1 & 0\\
    \hline
    1 & 2 & 2 & 2\\
    \hline
    2 & 2 & 0 & 1\\
    \hline
    \end{tabular}
    \caption{
    Truth table for the projection gates $\projone, \projtwo$. 
    Note that \(\projone, \projtwo\) are symmetric, i.e. \(\projone(x, y) = \projone(y, x)\) and \(\projtwo(x, y) = \projtwo(y, x)\).
    }
  \label{tab:ternary-q3-k2-projectors}
  \end{table}
  \subsection{Denoising threshold}
  \label{ssec:boolean-denoising-threshold}
  Consider a denoising gate \(\den: [3]^2 \to [3]\) with truth table given in \cref{tab:binary-q3-k2-denoiser}, which performs denoising for two logical states over a ternary alphabet, encoded using the physical alphabet characters `\(0\)' and `\(1\)'.
  The third physical alphabet character, `\(2\)', can be understood as an error flag, while the first two physical alphabet characters constitute the logical alphabet.
  This \(\den\) gate is designed to be used to combine multiple redundant computations to correct for any possible errors. Therefore, we can think of this gate as having realizations from two independent and identically distributed random variables as input, where the mode of the random variable is the value the inputs would take if no errors occurred in the computation. We would like the denoising gate to output this ``correct'' value as often as possible. If both inputs to the gate are the same member of the logical alphabet, then it is most likely that both inputs are correct and therefore the gate should output the value of the inputs, while if the inputs are differing members of the logical alphabet we output a `\(2\)' as it is equally likely that an error has occurred for each of the inputs. 
  If one input is a `\(2\)' while the other is a member of the logical alphabet, we output the input that is a member of the logical alphabet, as it is more likely than not that the latter input is correct.
  Finally, if both inputs are `\(2\)'s, we output a `\(2\)' as we do not have any information about which member of the logical alphabet is correct.

  \begin{table}[htb]
    \centering
    \setlength{\tabcolsep}{5pt}
    \begin{tabular}{|c|c|c|}
    \hline
    $x$ & $y$ & \(\den(x,y)\) \\
    \hline
    0 & 0 & 0 \\
    \hline
    0 & 1 & 2 \\
    \hline
    0 & 2 & 0 \\
    \hline
    1 & 1 & 1 \\
    \hline
    1 & 2 & 1 \\
    \hline
    2 & 2 & 2 \\
    \hline
    \end{tabular}
    \caption{
    Truth table for a balanced denoising gate \(\den\) for binary computation over a ternary alphabet. 
    Note that \(\den\) is symmetric, i.e. \(\den(x, y) = \den(y, x)\).
    }
  \label{tab:binary-q3-k2-denoiser}
  \end{table}

  We characterize the behavior of repeatedly applying \(\den\) to a random variable with probability distribution \(p\) over the alphabet \(\{0, 1, 2\}\).
  We parameterize the probability distribution over a ternary-valued random variable \(X\) with the tuple \((p_0, p_1)\) as follows:
  \begin{equation}
    \begin{pmatrix}
      \Pr[X = 0] \\
      \Pr[X = 1] \\
      \Pr[X = 2]
    \end{pmatrix}
    =
    \begin{pmatrix}
      p_0 \\
      p_1 \\
      1 - p_0 - p_1
    \end{pmatrix}.
  \end{equation}
  Let \(\den_\varepsilon\) denote an \(\varepsilon\)-noisy version of \(\den\).
  Then for two identical and independently distributed random variables \(X_1\) and \(X_2\) with distribution given by \((p_0, p_1)\), the distribution of \(\den_\varepsilon(X_1, X_2)\) is given by the vector-valued function
  \begin{equation}
    \mathcal{D}_\varepsilon(p_0, p_1) 
    \equiv
    \begin{pmatrix}
      \Pr[\den_\varepsilon(X_1, X_2) = 0] \\
      \Pr[\den_\varepsilon(X_1, X_2) = 1]
    \end{pmatrix}
    =
    \begin{pmatrix}
      \left(1 - \frac{3\varepsilon}{2}\right) (2 - p_0 - 2 p_1) p_0 + \frac{\varepsilon}{2} \\
      \left(1 - \frac{3\varepsilon}{2}\right) (2 - 2 p_0 - p_1) p_1 + \frac{\varepsilon}{2}
    \end{pmatrix}.
  \end{equation}

  \begin{figure}
    \centering
    \includegraphics[width=0.99\textwidth]{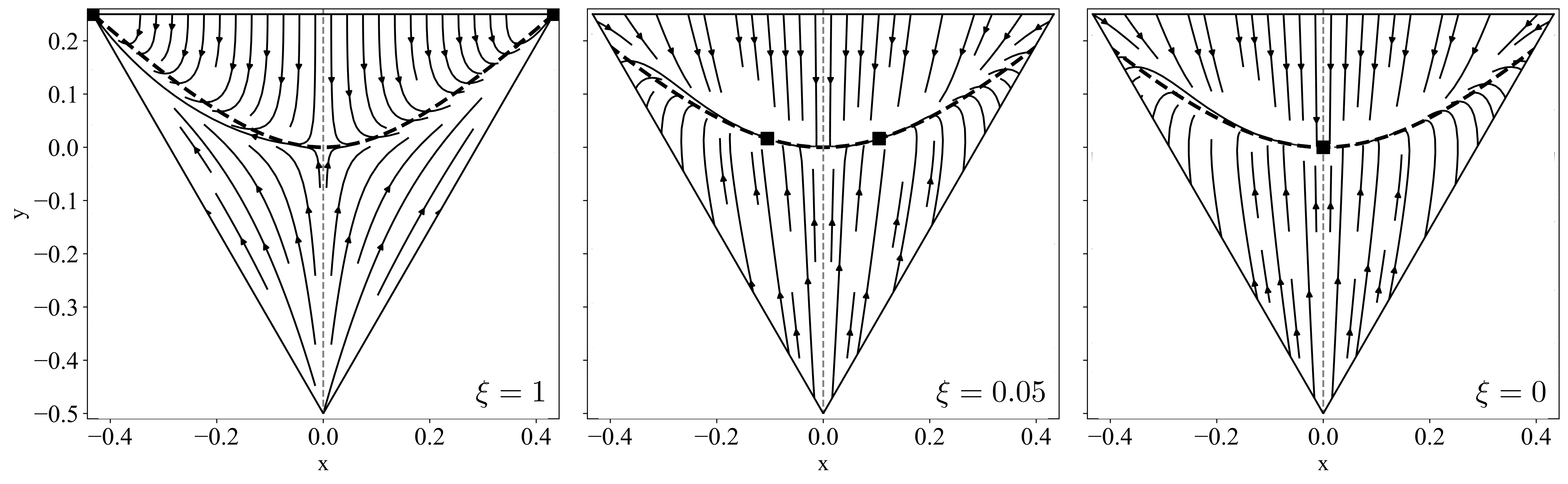}
    \caption{
    Streamlines of the vector field \(\bar{\mathcal{D}}_\varepsilon(x, y) - (x, y)\) for three values of \(\xi\) showing the \(y\) nullcline (dashed black line), stable fixed-points \((x_\pm, y_\pm)\) (square markers), and line separating logical states (dashed grey line).
    }
  \label{fig:denoising-stream-plot}
  \end{figure}

  \begin{lemma} 
  \label{lem:fixed-point-convergence}
    Given a random variable with an initial probability distribution \((p_0, p_1, 1 - p_0 - p_1)\) over \((0, 1, 2)\), we observe that
    \begin{equation}
      \lim_{n \to \infty} \mathcal{D}_\varepsilon^n(p_0, p_1) = 
      \begin{cases}
        \left(\frac{1}{3}, \frac{1}{3}\right)  & \text{if } p_0=p_1, \\
        \left(\frac{(1 - 3 \varepsilon) + \Delta}{2 - 3\varepsilon}, \frac{(1 - 3 \varepsilon) - \Delta}{2 - 3\varepsilon}\right)  & \text{if } p_0 > p_1, \\
        \left(\frac{(1 - 3 \varepsilon) - \Delta}{2 - 3\varepsilon}, \frac{(1 - 3 \varepsilon) + \Delta}{2 - 3\varepsilon}\right)  & \text{if } p_0 < p_1,
      \end{cases}
    \end{equation}
    where \(\Delta = \sqrt{(1 - 6 \varepsilon)(1 - 2 \varepsilon)}\).
  \end{lemma}
  \begin{proof}
    It is straightforward to verify that the limit distributions described above correspond to fixed-points of \(\mathcal{D}_\varepsilon\), and that they are the only fixed-points within the allowable region.

    We first consider the case when \(p_0 = p_1 = p\). 
    By induction, any number of applications of \(\den_\varepsilon\) will maintain equality in the first and second arguments as
    \begin{equation}
      \mathcal{D}_\varepsilon(p, p) 
      =
      \begin{pmatrix}
        p' \\
        p'
      \end{pmatrix},
    \end{equation}
    where \(p' = \left(1 - \frac{3\varepsilon}{2}\right) (2 - 3 p) p + \frac{\varepsilon}{2}\).

  Now note that for \(0 < \varepsilon \le 1/6\) and \(p \in [0, 1/2)\),
  \begin{align*}
    \abs*{p' - \frac{1}{3}} &= \abs*{\left(1 - \frac{3\varepsilon}{2}\right) (2 - 3 p) p + \frac{\varepsilon}{2} - \frac{1}{3}} \\
    &= \abs*{3 \left(1 - \frac{3\varepsilon}{2}\right) \parens*{p - \frac{1}{3}}^2} \\
    &= \abs*{3 \left(1 - \frac{3\varepsilon}{2}\right) \parens*{p - \frac{1}{3}}} \abs*{p - \frac{1}{3}} \\
    & < \abs*{p - \frac{1}{3}}.
  \end{align*}
  Therefore, repeated iteration results in convergence to \(1/3\).
  We conclude that \(\lim_{n \rightarrow \infty} \mathcal{D}_\varepsilon^n(p, p) = (1/3, 1/3)\) for all \(p \in [0, 1/2]\).

    Next, we make use of the parameterization of \cref{eq:xy-reparameterization}.
    In these coordinates, the action of the denoising gate is
    \begin{equation}
    \label{eq:denoising-iteration-xy}
      \bar{\mathcal{D}}_\varepsilon(x, y) 
      \equiv
      (\xi + 3)
      \begin{pmatrix}
        \frac{x(1 - y)}{3} \\
        \frac{x^2 - y^2}{2}
      \end{pmatrix},
    \end{equation}
    where we introduce \(\xi = 1 - 6\varepsilon > 0\) to make computations simpler.
    We consider the dynamics of probability distributions over the ternary alphabet under repeated application of \cref{eq:denoising-iteration-xy} with sequences \(\{(x_i, y_i)\}_{i=0}^{\infty}\), such that \((x_{i+1}, y_{i+1}) = \bar{\mathcal{D}}_\varepsilon(x_i, y_i)\).
    These discrete dynamics can be more easily visualized using a continuous approximation by plotting the streamlines of the vector field \(\bar{\mathcal{D}}_\varepsilon(x, y) - (x, y)\) as in \cref{fig:denoising-stream-plot}.
    
    Our original problem is equivalent to showing that all points \((x_0, y_0)\) in the regions
    \begin{align}
      \mathcal{R}^{(0)} = \left\{(x, y) \middle| -\frac{\sqrt{3}}{4} \le x < 0, -\frac{1}{2} \le y \le \frac{1}{4}, 0 \le 1 + 2 \sqrt{3} x + 2 y \right\}, \\
      \mathcal{R}^{(1)} = \left\{(x, y) \middle| 0 < x \le \frac{\sqrt{3}}{4}, -\frac{1}{2} \le y \le \frac{1}{4}, 0 \le 1 - 2 \sqrt{3} x + 2 y \right\},
    \end{align}
    converge to their respective fixed-points,
    \begin{equation}
      \left(x^{(0)}, y^{(0)}\right) = \left(-\frac{\sqrt{\xi(\xi + 2)}}{\xi + 3}, \frac{\xi}{\xi + 3}\right),
      \qquad\text{and}\qquad
      \left(x^{(1)}, y^{(1)}\right) = \left(\frac{\sqrt{\xi(\xi + 2)}}{\xi + 3}, \frac{\xi}{\xi + 3}\right),
    \end{equation}
    for \(0 \le \xi < 1\).

  We focus on the case \(p_0 < p_1 \implies (x_0, y_0) \in \mathcal{R}^{(1)}\), with \((x_0, y_0) \in \mathcal{R}^{(0)}\) following analogously.
  First, note that \(x > 0\) is preserved under iteration of \cref{eq:denoising-iteration-xy}, and therefore $(x_0, y_0) \in \mathcal{R}^{(1)}$ implies that for all $i > 0$, $(x_i, y_i) \in \mathcal{R}^{(1)}$.
  One may verify using cylindrical algebraic decomposition that the following is a Lyapunov function for the dynamics \cite{sm}:
  \begin{equation}
    V_\xi(x, y) = \left((\xi+3)^2x^2-(\xi+2)(\xi+3)y\right)^2+\left((\xi+3)^2x^2-\xi(\xi+2)\right)^2.
  \end{equation}
  That is, for all $i \ge 0$, $V_\xi(x_{i+1}, y_{i+1}) \le V_\xi(x_i, y_i)$, and the function is strictly minimized at \((x^{(1)}, y^{(1)})\) for all points in \(\mathcal{R}^{(1)}\).
\end{proof}

  The purpose of \cref{lem:fixed-point-convergence} is to show that when $\varepsilon < 1/6$, there exist stable fixed points that are distinguishable from each other that repeated denoising will converge to. (For $\varepsilon \ge 1/6$, this is not the case, and repeatedly applying $D_\varepsilon$ results in convergence to $(\frac{1}{3}, \frac{1}{3})$ regardless of the initial probability distribution.) Then, we can associate these fixed points in the regions \(\mathcal{R}^{(0)}\), \(\mathcal{R}^{(1)}\) with the logical \(0\) and \(1\) state respectively when performing computations.

  \subsection{Computation threshold}
  \label{ssec:boolean-computation-threshold}
  Next, we show that universal boolean computation is possible up to this threshold of \(\varepsilon < 1 / 6\) by providing a universal computation gate that is compatible with the denoising gate~$\den$. We need to design a computation gate that is 
  universal for binary computation over a ternary alphabet. Because the $\nand$~gate is universal over a binary alphabet, we will design a balanced ternary gate whose behavior when applied to the fixed-points given in~\cref{lem:fixed-point-convergence} is analogous to the behavior of a $\nand$~gate.

  To design this gate, we first note that the restriction of the gate to the logical alphabet \(\{0, 1\}\) should be the $\nand$~gate. 
  It remains to assign outputs for cases when one or both inputs are \(2\). In order to keep the gate balanced, we therefore must assign the output
  \(2\) when both inputs are \(2\), and need to determine whether to assign an output of \(0\) or \(2\) when one input is a \(0\) or a \(1\) and the other is a \(2\).

  Consider the $\errornand$~gate, a generalization of the binary $\nand$~gate to a ternary alphabet where one element is used for error signaling with truth table given in \cref{tab:binary-q3-k2-enand}. 

  \begin{table}[htb]
    \centering
    \setlength{\tabcolsep}{5pt}
    \begin{tabular}{|c|c|c|}
    \hline
    $x$ & $y$ & $\errornand(x,y)$ \\
    \hline
    0 & 0 & 1 \\
    \hline
    0 & 1 & 1 \\
    \hline
    0 & 2 & 2 \\
    \hline
    1 & 1 & 0 \\
    \hline
    1 & 2 & 0 \\
    \hline
    2 & 2 & 2 \\
    \hline
    \end{tabular}
    \caption{Truth table for a universal $\errornand$~for binary computation over a ternary alphabet.
    Note that \(\errornand\) is symmetric, i.e. \(\errornand(x, y) = \errornand(y, x)\).
    }
    \label{tab:binary-q3-k2-enand}
  \end{table}

  We verify that this gate applied to the fixed-points given in~\cref{lem:fixed-point-convergence} behaves analogously to a $\nand$~gate acting on the corresponding elements of the logical alphabet.

  \begin{lemma}
  \label{lem:reliable-boolean-computation-over-ternary-alphabet}
    For \((k = 2)\)-input gates over an alphabet of size \(q = 3\), reliable universal boolean computation using $\varepsilon$-noisy $\den_\varepsilon$ and $\errornand_\varepsilon$ gates is possible for \(\varepsilon < 1 / 6\).
  \end{lemma}
  \begin{proof}
    We associate the logical values \(\{0, 1\}\) to the distributions
    \begin{equation}
      \label{eq:def-den-fixed-points}
      P_0 = 
      \begin{pmatrix}
        p_+ \\
        p_- \\
        1 - p_+ - p_-
      \end{pmatrix},
      \qquad\text{and}\qquad
      P_1 = 
      \begin{pmatrix}
        p_- \\
        p_+ \\
        1 - p_+ - p_-
      \end{pmatrix},
    \end{equation}
    where
    \begin{equation}
      \label{eq:def-pplus-pminus}
      p_\pm = \frac{(1 - 3 \varepsilon) \pm \Delta}{2 - 3\varepsilon},
      \qquad
      \Delta = \sqrt{(1 - 6 \varepsilon)(1 - 2 \varepsilon)}.
    \end{equation}
    Note that we are able to generate signals with distributions arbitrarily close to \(P_0\) and \(P_1\) through repeated denoising via \cref{lem:fixed-point-convergence}.

    Once again, we parameterize distributions over the ternary alphabet by using the probabilities corresponding to the two members of the logical alphabet, and let \(\mathcal{G}_\varepsilon: (P_{\text{in}, 0}, P_{\text{in}, 1}) \mapsto P_{\text{out}}\) denote the function mapping input probability distributions to output probability distributions under the \(\varepsilon\)-noisy $\errornand_\varepsilon$~gate;
    further, we let \(\mathcal{G}_\varepsilon^{(0)}\) and \(\mathcal{G}_\varepsilon^{(1)}\) denote the respective probabilities of the output being in the \(0\) or \(1\) state.
    For the computation under the $\errornand_\varepsilon$~gate to be successful, given inputs from \(\{P_0, P_1\}\), the output must land in the appropriate half of the simplex. Note that as in \cref{ssec:projecting-gates} it is sufficient to prove these inequalities for the gates without having added symmetric noise. In other words, for inputs \(u, v \in \{0, 1\}\), \(\mathcal{G}(P_u, P_v) \in \mathcal{R}^{(\nand(u, v))}\), or more explicitly
    \begin{subequations}
    \label{subeqs:enand-inequalities}
      \begin{align}
        \mathcal{G}^{(0)}(P_0, P_0) &< \mathcal{G}^{(1)}(P_0, P_0) \\
        \mathcal{G}^{(0)}(P_0, P_1) &< \mathcal{G}^{(1)}(P_0, P_1) \\
        \mathcal{G}^{(0)}(P_1, P_0) &< \mathcal{G}^{(1)}(P_1, P_0) \\
        \mathcal{G}^{(0)}(P_1, P_1) &> \mathcal{G}^{(1)}(P_1, P_1),
      \end{align}
    \end{subequations}
    thus implementing the binary \nand~gate.
    We find that
    \begin{subequations}
      \begin{align}
        \mathcal{G}(P_0, P_0) &= \left(\frac{2-8\varepsilon + 3\varepsilon^2 - 2\Delta}{(2-3\varepsilon)^2},\frac{2-10\varepsilon+15\varepsilon^2+(2-6\varepsilon)\Delta}{(2-3\varepsilon)^2}\right) \\
        \mathcal{G}(P_0, P_1) &= \left(\frac{\varepsilon(8-21\varepsilon)}{(2-3\varepsilon)^2},\frac{4-26\varepsilon+39\varepsilon^2}{(2-3\varepsilon)^2}\right) \\
        \mathcal{G}(P_1, P_0) &= \left(\frac{\varepsilon(8-21\varepsilon)}{(2-3\varepsilon)^2},\frac{4-26\varepsilon+39\varepsilon^2}{(2-3\varepsilon)^2}\right) \\
        \mathcal{G}(P_1, P_1) &= \left(\frac{2-8\varepsilon + 3\varepsilon^2 + 2\Delta}{(2-3\varepsilon)^2},\frac{2-10\varepsilon+15\varepsilon^2-(2-6\varepsilon)\Delta}{(2-3\varepsilon)^2}\right).
      \end{align}
    \end{subequations}
    Using these expressions, it is straightforward to verify the inequalities in \cref{subeqs:enand-inequalities} for \(\varepsilon < 1/6\).

    Since \(\mathcal{G}_\varepsilon\) is a continuous function of the input probability distributions, we have that for any \(\varepsilon < 1 / 6\) there exists some $\epsilon$-radius ball around the fixed-points of \(\mathcal{D}_\varepsilon\) such that $\errornand_\varepsilon$~can be used to perform computation. (Recall that as long as the output of $\errornand_\varepsilon$ is in the correct half of the probability simplex, we will be able to denoise to the correct fixed point by \cref{lem:fixed-point-convergence}). Consider some output random variable $X_0$ of a $\errornand_\varepsilon$ gate that we would like to denoise. For $i = 0$ to some value $n \in \mathbb{Z}$, we will prepare a copy $X_i'$ that is independent of $X_i$ but has the same distribution over $\{0,1,2\}$, and then let $X_{i+1}=\den_\varepsilon(X_i,X_i')$. 
    By \cref{lem:fixed-point-convergence}, there exists some number $n$ such that repeatedly applying this operation $n$ times results in some random variable $X_n$ such that the distribution of $X_n$ is within a $\epsilon$-radius ball centered at a fixed point of $\mathcal{D}_\varepsilon$. Applying another $\errornand_\varepsilon$ gate to $X_n$ will result in a distribution that is in the correct half of the probability simplex, and can be denoised to the correct logical value using $\den_\varepsilon$.
    This construction of alternating $\errornand_\varepsilon$~gates followed by sufficiently many $\den_\varepsilon$~gates therefore allows arbitrary binary computation to be reliably performed for \(\varepsilon < 1 / 6\).
  \end{proof}

  \subsection{Lifting gates}
  \label{ssec:lifting-gates}
  Finally, we can lift the binary output of the computation to the ternary alphabet. To do this, we use a lifting gate \(\lift: [3]^2 \to [3]\) with truth table given in \cref{tab:binary-q3-k2-lift}.

  \begin{table}[htb]
    \centering
    \setlength{\tabcolsep}{5pt}
    \begin{tabular}{|c|c|c|}
    \hline
    $x$ & $y$ & $\lift(x,y)$ \\
    \hline
    0 & 0 & 0 \\
    \hline
    0 & 1 & 1 \\
    \hline
    0 & 2 & 0 \\
    \hline
    1 & 0 & 2 \\
    \hline
    1 & 1 & 1 \\
    \hline
    1 & 2 & 2 \\
    \hline
    2 & 0 & 2 \\
    \hline
    2 & 1 & 1 \\
    \hline
    2 & 2 & 0 \\
    \hline
    \end{tabular}
    \caption{Truth table for a universal $\lift$~for lifting two logical bits into a ternary output.
    Note that \(\lift\) is  \textit{not} symmetric in the two inputs.
    }
    \label{tab:binary-q3-k2-lift}
  \end{table}
  We must verify that $\lift_\varepsilon$ applied to the fixed points of $\mathcal{D}_\varepsilon$ will result in a random variable that is more likely to take on the value of the correct physical alphabet element than either of the incorrect physical alphabet elements.
  \begin{lemma}
    \label{lem:lifting-gate-correctness}
    Given inputs $i,j$ with $(i,j) \in \{(0,0), (0,1), (1,0)\}$, we have that $\lift_\varepsilon(P_i, P_j)$ is more likely to output $2i+j$ than either other value.
  \end{lemma}
  \begin{proof}
    It suffices to show that this is true for $\lift$, as applying symmetric noise will not affect which output probability is greatest as seen in \cref{ssec:projecting-gates}. Let $\mathcal{F}:(P_{\text{in},0}, P_{\text{in},1}) \to P_\text{out}$ denote the function mapping input probability distributions to output probability distributions under $\lift$, and let $\mathcal{F}^{(i)}$ be the probability of the output being in state $i$. It suffices to show that
  \begin{subequations}
    \label{subeqs:lift-inequalities}
      \begin{align}
        \mathcal{F}^{(0)}(P_0, P_0) &> \mathcal{F}^{(1)}(P_0, P_0),\mathcal{F}^{(2)}(P_0, P_0) \label{eq:lift_ineq_1} \\
        \mathcal{F}^{(1)}(P_0, P_1) &> \mathcal{F}^{(0)}(P_0, P_1), \mathcal{F}^{(2)}(P_0, P_1) \label{eq:lift_ineq_2} \\
        \mathcal{F}^{(2)}(P_1, P_0) &> \mathcal{F}^{(0)}(P_1, P_0), \mathcal{F}^{(1)}(P_1, P_0) \label{eq:lift_ineq_3},
      \end{align}
    \end{subequations}
    for all $0 < \varepsilon < 1/6$.
    We find that
    \begin{subequations}
      \begin{align}
        \mathcal{F}(P_0, P_0) &= \left(\frac{2-11\varepsilon+21\varepsilon^2+(2-3\varepsilon)\Delta}{(2-3\varepsilon)^2},\frac{1-3\varepsilon-\Delta}{2-3\varepsilon}, \frac{8\varepsilon-21\varepsilon^2}{(2-3\varepsilon)^2}\right) \\
        \mathcal{F}(P_0, P_1) &= \left(\frac{5\varepsilon-3\varepsilon^2+3\varepsilon\Delta}{(2-3\varepsilon)^2},\frac{1-3\varepsilon+\Delta}{2-3\varepsilon},\frac{2-8\varepsilon+3\varepsilon^2-2\Delta}{(2-3\varepsilon)^2}\right) \\
        \mathcal{F}(P_1, P_0) &= \left(\frac{5\varepsilon-3\varepsilon^2-3\varepsilon\Delta}{(2-3\varepsilon)^2},\frac{1-3\varepsilon-\Delta}{2-3\varepsilon},\frac{2-8\varepsilon+3\varepsilon^2+2\Delta}{(2-3\varepsilon)^2}\right),
      \end{align}
    \end{subequations}
    and it is straightforward to verify the inequalities in \cref{subeqs:lift-inequalities} for \(\varepsilon < 1/6\).
  \end{proof}

\subsection{Proving \cref{thm:universal-ternary-computation}}
\label{ssec:proving_ternary_computation}
  We now combine $\{\projone_\varepsilon, \projtwo_\varepsilon, \den_\varepsilon, \errornand_\varepsilon, \lift_\varepsilon\}$ to perform universal computation over a ternary alphabet.
  \begin{proof}[Proof of \cref{thm:universal-ternary-computation}]
    Given any ternary computation that we wish to perform, we can first convert the input to binary using the $\projone_\varepsilon$ and $\projtwo_\varepsilon$ gates, giving the binary representation of each input. By \cref{lem:fixed-point-convergence}, we can use $\den_\varepsilon$ to denoise arbitrarily close to the fixed points of $\mathcal{D}_\varepsilon$. By \cref{lem:reliable-boolean-computation-over-ternary-alphabet}, we can then reliably compute the most and least significant bits of the desired ternary output by using $\errornand_\varepsilon$ and $\den_\varepsilon$. Then, we can use $\lift_\varepsilon$ to convert the bits into the desired ternary output via \cref{lem:lifting-gate-correctness}, completing the computation.
  \end{proof}

\section{Concluding Remarks}
\label{sec:conclusion}
  In this paper, we extended a number of positive results for reliable Boolean (\(q = 2\)) computation to the setting of alphabets of size \(q > 2\).

  In \cref{sec:large-alphabet-computation}, we provided a generalization of the results of Evans and Schulman \cite{evans2003on-the-maximum} for alphabets of size \(q > 2\), showing a positive result for the reliable computation.
  Notably, we find that reliable universal computation is possible up to \(\varepsilon < (q - 1) / q\) as \(k \rightarrow \infty\).
  While this result applies only in the case of \(q\)-ary symmetric noise, we note that similar asymptotic arguments apply to generic models of i.i.d. gate noise over alphabets of size \(q\).
  For example, suppose that the probability of maintaining the correct logical state is \(1 - \epsilon\) and the probability of the most likely erroneous state is \(c \epsilon\) for some \(c > 0\). 
  Then the asymptotic \(k \rightarrow \infty\) threshold occurs at the point where the two are equal, i.e. \(\epsilon = 1 / (1 + c)\), which is maximized in the case of symmetric noise \(c = 1 / (q - 1)\).
  General results for finite \(k\) depend on the existence of a properly adapted majority gate and computation gate.
  In the case of \(q\)-ary symmetric noise, we showed positive computation results for finite fan-in \(k = 3\) and \(q\) prime up to the transcritical bifurcation (\cref{lem:lower-bound-on-majqk-denoising-threshold}).
  We conjecture that reliable computation is possible up to the transcritical bifurcation for all \(k\) odd and \(q\) prime.
  Interestingly, and in contrast with the Boolean case \cite{evans2003on-the-maximum}, distinguishable stable fixed-points exist past the transcritical bifurcation (see the \(q = k = 3\) example in \cref{fig:fixed-point-diagram,fig:denoising-stream-plot-symmetric}), leading to the possibility of reliable computation beyond the lower-bound of \cref{lem:lower-bound-on-majqk-denoising-threshold}, though a redefinition of the logical regions in \(\Delta_q\) is required.

  In \cref{sec:universal-ternary-computation}, we extended the work of Evans and Pippenger \cite{evans1998on-the-maximum} for computation using \(2\)-input gates, showing that universal computation is possible for \(\varepsilon < 1 / 6\) using error signaling.
  We would also like to highlight a conceptual difference in the way fault-tolerance comes about for the model of Boolean computation over a ternary alphabet in \cref{sec:universal-ternary-computation}.
  The \den~(\cref{tab:binary-q3-k2-denoiser}) and \errornand~(\cref{tab:binary-q3-k2-enand})~gates in effect use the third element of the physical alphabet as an error flag, signaling uncertainty to subsequent gates, resulting in a nearly two-fold increase in the denoising threshold compared to the analogous result over Boolean alphabets \cite{evans1998on-the-maximum}.
  The use of this third alphabet character is reminiscent of the usage of flag qubits \cite{chao2020flag} and space-time codes \cite{gottesman2022opportunities} in quantum computation, where similar error signaling mechanisms are used to achieve quantum fault-tolerance.
  There are many extensions to the error signaling construction developed in this paper that would be interesting to explore. 
  It would also be interesting to study the optimal logical alphabet size for a given error model over the physical alphabet, and how to design error signaling gates to achieve optimal threshold error rates.

  Finally, increasing the physical alphabet size under the \(q\)-ary symmetric noise model increases the threshold error rate --- for example, the result of \cref{sec:universal-ternary-computation} shows that reliable Boolean computation is possible for \(\varepsilon < \beta = 1 / 6\), nearly double the nominal threshold for standard Boolean computation \cite{hajek1991on-the-maximum,evans1998on-the-maximum}.
  However, we are not advocating generically for performing computation over larger alphabets as nominal error rates are not strictly comparable between different alphabet sizes.
  Since all errors are equally likely by definition under the symmetric noise model, signals are distinguishable for \(\varepsilon < (q - 1) / q\);
  and since distinguishablility is sufficient for reliable computation as \(k \rightarrow \infty\), computation over larger physical alphabets is afforded a natural advantage under the assumption of symmetric noise.
  In practice, error rates need to be grounded in the utilization of some resource \cite{impens2004fine-grained,thaker2005recursive,thaker2008on-using,tan2023resource};
  for example, one could compare the energy required to implement an \(\varepsilon\)-noisy binary circuit with that required to implement an \(\varepsilon\)-noisy ternary circuit.
  The limit of large \(q\) and large \(k\), for which our results yield the most favorable nominal thresholds, comes at the cost of significant resource requirements which need to be taken into account in practice.
  The comparison of error rates across different physical alphabet sizes presents an interesting opportunity for future work.

\acknowledgements{
The authors acknowledge support from the Institute for Artificial Intelligence and Fundamental Interactions (IAIFI) through NSF Grant No. PHY-2019786; and the U.S. Department of Energy, Office of Science, National Quantum Information Science Research Centers, Co-design Center for Quantum Advantage (C2QA) under contract number DE-SC0012704.
}


\nocite{Mathematica}

\bibliographystyle{IEEEtran}
\bibliography{refs}


\clearpage

\appendix

\end{document}